\documentclass[11pt]{article}
\usepackage{amssymb, amsmath, amsthm, amsfonts, bbm}
\usepackage{algorithm, algpseudocode}

\usepackage{graphicx,caption,subcaption}
\usepackage{natbib}
\usepackage{color}
\usepackage{xcolor}
\usepackage[colorlinks,linkcolor=blue,citecolor=blue!70!black,urlcolor=black]{hyperref}

\usepackage{fullpage}

\usepackage{booktabs} 
\usepackage{tcolorbox}
\tcbuselibrary{most}

\usepackage{tikz}
\usetikzlibrary{arrows.meta, positioning}
\usepackage{fontawesome5}
  \usetikzlibrary{shapes.geometric, positioning, fit, calc}
  
\usepackage{booktabs} 
\usepackage{tcolorbox}
\tcbuselibrary{most}

\usepackage{bbm}

\usepackage{mathtools}

\def\eqref#1{equation~\ref{#1}}

\def\1{\bm{1}}

\def\rva{{\mathbf{a}}}

\def\rve{{\mathbf{e}}}

\def\rvr{{\mathbf{r}}}

\def\rvx{{\mathbf{x}}}

\def\rvz{{\mathbf{z}}}

\def\rmA{{\mathbf{A}}}
\def\rmB{{\mathbf{B}}}
\def\rmC{{\mathbf{C}}}

\def\rmI{{\mathbf{I}}}

\def\rmR{{\mathbf{R}}}

\DeclareMathAlphabet{\mathsfit}{\encodingdefault}{\sfdefault}{m}{sl}
\SetMathAlphabet{\mathsfit}{bold}{\encodingdefault}{\sfdefault}{bx}{n}

\def\gA{{\mathcal{A}}}

\def\gC{{\mathcal{C}}}

\def\gM{{\mathcal{M}}}

\def\gR{{\mathcal{R}}}

\DeclareMathOperator*{\argmax}{arg\,max}

\newcommand{\N}{\mathbb{N}}

\newcommand{\calM}{\mathcal{M}}

\newcommand{\calH}{\mathcal{H}}
\newcommand{\calA}{\mathcal{A}}

\newcommand{\lct}{\ell_{t,c}}

\newcommand{\Gct}{G_{t,c}}

\newcommand{\gMg}{\gM_{\mathrm{greedy}}}
\newcommand{\gMuni}{\gM_{\mathrm{uni}}}
\newcommand{\gMrr}{\gM_{\mathrm{RR}}}

\usepackage{amsthm}
\newtheorem*{theorem*}{Theorem}
\newtheorem{theorem}{Theorem}
\newtheorem{definition}{Definition}
\newtheorem{condition}{Condition}
\newtheorem{proposition}[theorem]{Proposition}
\newtheorem{corollary}[theorem]{Corollary}
\newtheorem{remark}{Remark}[section]
\newtheorem{lemma}[theorem]{Lemma}

\usepackage[capitalize,noabbrev]{cleveref}
\crefname{theorem}{Theorem}{Theorems}
\crefname{lemma}{Lemma}{Lemmas}
\crefname{corollary}{Corollary}{Corollaries}
\crefname{proposition}{Proposition}{Propositions}
\crefname{definition}{Definition}{Definitions}
\crefname{example}{Example}{Examples}
\crefname{remark}{Remark}{Remarks}
\crefname{assumption}{Assumption}{Assumptions}
\crefname{fact}{Fact}{Facts}
\crefname{result}{Result}{Results}
\crefname{problem}{Problem}{Problems}
\crefname{question}{Question}{Questions}

  {\end{tcolorbox}}
 
  {\end{tcolorbox}}

\title{Plausible Deniability Guarantees for Whistleblowers}

\author{
  Leo Richter$^{1}$\: and Matt J. Kusner$^{2,3}$
}
\date{
$^{1}$\textit{University College London}\\
$^{2}$\textit{École Polytechnique de Montréal}\\
$^{3}$\textit{Mila - Québec Artificial Intelligence Institute}\\
}

\newcounter{daggerfootnote}

\begin{document}

\onecolumn
\maketitle

\begin{abstract}

Whistleblowers are a key safeguard against organizational wrongdoing, but the threat of retaliation deters reporting. Existing whistleblower-protection proposals lack formal privacy guarantees, and existing differential privacy mechanisms do not directly target the natural threat model — one in which the audited organization itself observes auditor selection decisions and uses them to identify reporters. We formalize protection against a strong-adversary threat model as per-report $(0, \delta)$-differential privacy on the transcript of audit selections. Within this framework we prove that a natural approach — randomized response applied at the selection step — can never outperform uniform random auditing by more than $\delta$ at any horizon. We then give a generic mechanism that reduces private auditing to private continual counting: any $(0, \delta)$-DP continual counter plugs in by post-processing, and the audit transcript inherits the same per-report guarantee. Instantiating the reduction with a recent work in continual counting yields per-report $(0, \delta)$-DP with noise scaling as $O(\sqrt{\log T})$ across a horizon of $T$ audit decisions. A utility theorem shows that the selection error vanishes whenever the noisy report gap between the most-reported organization and the runner-up grows faster than $\sqrt{\log T}$. Simulations show a substantial improvement over randomized response. 
\end{abstract}

\setcounter{tocdepth}{0}

\section{Introduction}
\label{sec:intro}
Whistleblowers are a crucial safeguard against organizational wrongdoing~\citep{meijer2024global}, from misconduct in private organizations \citep{dyck2010who,stubben2020evidence} to
national governments \citep{scheuerman2014whistleblowing}. Employees, who account for a significant fraction of whistleblower reports~\citep{dyck2010who}, frequently face retaliation from the organizations they expose \citep{lee2011whistleblower}. 
This is a major deterrent to whistleblowing: perceived threat of retaliation is negatively associated with the intention to blow the whistle~\citep{near1986retaliation, mesmer2005whistleblowing,cassematis2013prediction}.

For this reason, confidentiality and anonymity are widely treated as central to effective whistleblower protection~\citep{sharma2018comparison,chassang2019crime,baljija2023evaluating}. This is also reflected in recent legal frameworks: the EU Whistleblower Protection Directive requires confidential reporting channels and protection from retaliation~\citep{EU20191937}, and the EU AI Act extends these protections to reports of AI Act infringements~\citep{ai_act_2024}. Yet these regimes typically specify confidentiality at the level of the reporting channel. They do not provide a formal guarantee about what an organization may infer from the auditor's subsequent actions.

This distinction matters because anonymity can fail even when reports are submitted confidentially. Suppose an employee reports a concern to an independent auditor, and the auditor soon afterwards investigates the employee's organization. The organization may not observe the report itself, but it observes the audit. If only a few employees had access to the relevant information, the audit can reveal who likely reported. Thus, the central privacy question is not only whether a reporting channel hides names, but whether the observable audit transcript preserves plausible deniability for each individual report.

This problem is particularly acute in AI governance. Many risks from advanced AI systems are likely to be observed first by employees and contractors inside the organizations developing them, before they are visible to external auditors, regulators, or the public. 
Recent proposals for frontier-AI regulation, responsible reporting, and end-to-end algorithmic auditing all rely on insider information as an important signal~\citep{rajiclosing2020, anderljung2023frontierairegulationmanaging,Kolt_Anderljung_Barnhart_Brass_Esvelt_Hadfield_Heim_Rodriguez_Sandbrink_Woodside_2024}. For such systems to be credible, a developer should not be able to infer from the timing or target of a model evaluation, compliance review, or safety audit that a particular employee filed a report.

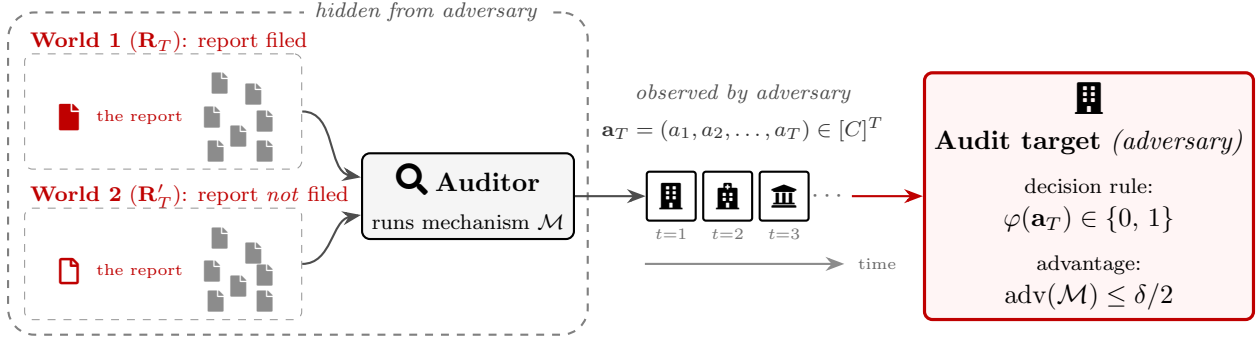
\begin{figure*}[t]
\centering
\resizebox{\textwidth}{!}{%
\begin{tikzpicture}[
    font=\small,
    >={Stealth[length=2.5mm]},
    flow/.style       ={->, thick, black!70},
    flowred/.style    ={->, thick, red!75!black},
    boundary/.style   ={draw=black!50, dashed, thick, rounded corners=10pt},
    starf/.style      ={star, star points=5, draw=red!75!black, fill=red!75!black,
                        minimum size=0.34cm, inner sep=0pt},
    starh/.style      ={star, star points=5, draw=red!75!black, fill=white,
                        line width=0.7pt, minimum size=0.34cm, inner sep=0pt},
    starg/.style      ={star, star points=5, draw=black!45, fill=black!45,
                        minimum size=0.18cm, inner sep=0pt},
    auditbox/.style ={draw, thick, rounded corners=3pt, fill=black!4, align=center,
                  minimum height=0.95cm, minimum width=2.1cm},
    advbox/.style     ={draw=red!75!black, very thick, rounded corners=3pt, fill=red!4,
                        align=center, minimum height=2.7cm, minimum width=3.7cm},
    aslot/.style ={draw, thick, rounded corners=1pt, minimum size=0.65cm,
               font=\small, fill=white, inner sep=0pt},
    worldframe/.style ={draw=black!40, dashed, rounded corners=3pt, thin, inner sep=4pt},  
repf/.style ={font=\small,      text=red!75!black, inner sep=0pt},
reph/.style ={font=\small,      text=red!75!black, inner sep=0pt},
repg/.style ={font=\scriptsize, text=black!45,    inner sep=0pt},
  ]

    
    

    \node[font=\scriptsize, text=red!75!black, anchor=west]
     at (-0.1, 2.10) {\textbf{World 1} ($\rmR_T$): report filed};        
    \node[repf] (tf) at (0.55, 1.10) {\faFile};
    \foreach \x/\y in {2.55/1.55, 3.0/1.40, 2.45/1.10, 3.15/1.05,
                       2.80/0.85, 3.15/0.65, 2.50/0.65}
      \node[repg] at (\x,\y) {\faFile};
    \node[worldframe, fit={(0.1,0.55) (3.50,1.80)}] (R1) {};           
    \node[font=\tiny, text=red!75!black, anchor=west] at (0.78, 1.10) {the report};
    
    \node[font=\scriptsize, text=red!75!black, anchor=west]
     at (-0.1, 0.05) {\textbf{World 2} ($\rmR'_T$): report \emph{not} filed}; 
    \node[reph] (ts) at (0.55,-0.95) {\faFile[regular]};
    \foreach \x/\y in {2.55/-0.50, 3.0/-0.70, 2.45/-0.95, 3.15/-1.00,
                       2.80/-1.15, 3.15/-1.35, 2.50/-1.35}
      \node[repg] at (\x,\y) {\faFile};
    \node[font=\tiny, text=red!75!black, anchor=west] at (0.78, -0.95) {the report};
    \node[worldframe, fit={(0.10,-1.45) (3.50,-0.25)}] (R2) {};         
    
    \node[boundary, fit={(-0.12,-1.65) (7.3,2.35)}, inner sep=4pt] (hidden) {}; 
    \node[font=\scriptsize\itshape, text=black!70, fill=white,
      inner xsep=4pt, inner ysep=0pt]
     at ($(hidden.north) + (1.7,0)$) {hidden from adversary};

    \node[auditbox] (M) at (5.85, 0.05)
     {{\large\faSearch}\;\textbf{Auditor}\\[2pt]
      {\scriptsize runs mechanism $\gM$}};

  \draw[flow] (R1.east) to[out=0, in=170] ([yshift=2.5mm]M.west);
  \draw[flow] (R2.east) to[out=0, in=190] ([yshift=-2.5mm]M.west);


  \node[font=\scriptsize, text=black!80, anchor=south]
       at (9.50, 0.65)
       {$\rva_T = (a_1, a_2, \ldots, a_T) \in [C]^T$};



    \node[aslot] (a1) at (8.55, 0.05) {\faBuilding};     
    \node[aslot] (a2) at (9.30, 0.05) {\faHospital};
    \node[aslot] (a3) at (10.05, 0.05) {\faUniversity};
    \node[font=\small, text=black!55] at (10.65, 0.05) {$\cdots$};
    
    \foreach \x/\t in {8.55/1, 9.30/2, 10.05/3}
      \node[font=\tiny, text=black!55] at (\x,-0.45) {$t{=}\t$};
    
    \draw[->, thick, black!45] (8.20,-0.85) -- (10.85,-0.85);
    \node[font=\tiny, text=black!60, anchor=west] at (10.90,-0.85) {time};

  \node[font=\scriptsize\itshape, text=black!70, anchor=south]
       at (9.50, 1.15) {observed by adversary};

  \draw[flow] (M.east) -- (a1.west);

  \node[advbox] (adv) at (14.10, 0.05)
       {{\large\faBuilding}\\[2pt]\textbf{Audit target} \textit{(adversary)}\\[4pt]
        {\scriptsize decision rule:}\\
        $\varphi(\rva_T)\in\{0,\,1\}$\\[4pt]
        {\scriptsize advantage:}\\
        $\mathrm{adv}(\gM)\le\delta / 2$};

  \draw[flowred] (10.95, 0.05) -- (adv.west);

\end{tikzpicture}}
\caption{\textbf{Threat model — two possible worlds.} In World 1, the report in question is filed (filled red star), yielding stream $\rmR_T$; in World 2, it is not (hollow icon), yielding $\rmR'_T$, with all other reports identical. The auditor processes the stream confidentially via mechanism $\gM$ and releases the audit transcript $\rva_T \in [C]^T$. The adversary — the audit target itself — observes only $\rva_T$ and applies a decision rule $\varphi : [C]^T \to \{0, 1\}$ to to decide whether the report was filed. The auditor's design bounds the adversary's advantage $\mathrm{adv}(\gM) \le \delta/2$.}
\label{fig:threat-model}
\vspace{-2ex}
\end{figure*}


The key insight to ensure this is that audit decisions need not be deterministic functions of report counts. If the auditor introduces randomness into the organizations selected for audit, the organization cannot be sure whether an audit was triggered by a report or by the auditor’s randomization. This gives whistleblowers plausible deniability: the ability to deny
responsibility due to other alternative possibilities. More specifically, we study the design of randomized audit policies that provide a quantifiable deniability guarantee. At each timestep, an auditor receives whistleblower report and selects an audit target for follow-up. The resulting sequence of audit decisions is observable to the audited organizations. Ideally, the mechanism should protect any single report while still allowing useful audit decisions, and should impose no bound on the number of reports received between audits.

One classic approach to achieve plausible deniability is using
randomized response \citep{warner1965randomized}: the mechanism sometimes follows the true signal and sometimes outputs a random response. Prior work \citep{chassang2019crime} adapts this idea to whistleblowing. Unfortunately, privacy pins randomized response to near-uniform auditing: at any fixed privacy level, it must select almost at random already at the first audit decision, regardless of the horizon.



\subsubsection*{Our contributions:}
\label{subsubsec:contributions}

\begin{itemize}
    \item \textbf{Formalization (Sections~\ref{sec:framework}--\ref{sec:threat_model}).} We formalize whistleblower auditing as per-report $(0, \delta)$-differential privacy on the transcript of audit selection decisions, under a strong-adversary threat model in which the audited organization itself is the adversary. To our knowledge, no prior work targets this exact granularity — per-report adjacency, per-decision observable, fixed horizon — under this threat model.
    \item \textbf{Negative result for randomized response (Section~\ref{sec:rand_response}, Theorem~\ref{thm:RR_tradeoff}).} We prove that randomized response applied at the selection step — the canonical adaptation of Warner's \cite{warner1965randomized} primitive used in the economics literature on whistleblowing \cite{chassang2019crime} — can never outperform uniform random auditing by more than $\delta$ at any horizons.
    \item \textbf{Generic reduction (Section~\ref{sec:mechanism}, Proposition~\ref{prop:privacy_reduction}).} We give a mechanism that reduces private auditing to private continual counting: any $(0, \delta)$-DP continual counter plugs in by post-processing, and the audit transcript inherits the same per-report guarantee. Future improvements in continual counting can be incorporated easily.
    \item \textbf{Concrete instantiation and utility analysis (Sections~\ref{sec:counter-instantiation} and~\ref{sec:utility}).} We instantiate the reduction with the Toeplitz-factorization counter of \citet{fichtenberger2023constant}, which was recently shown to be near-optimal for continual counting by \citet{dvijotham2024efficient}. This obtains per-report $(0, \delta)$-DP with noise scaling as $O(\sqrt{\log T})$. A utility theorem (Theorem~\ref{thm:utility}) shows the error decays whenever the noisy report-count gap between the organization with the most reports and the runner-up organization grows faster than $\sqrt{\log T}$ — an asymptotic improvement over randomized response.
    \item \textbf{Empirical validation (Section~\ref{sec:experiments}).} Simulations show that the resulting mechanism substantially outperforms randomized response in both a static gap sweep and a dynamic online auditing setting.
\end{itemize}

\section{An Auditing Framework}
\label{sec:framework}
We propose the following auditing setup. Fix an audit horizon $T \in \mathbb{N}$ and a number of audit
targets $C \in \mathbb{N}$.\footnote{The audit unit $c$ need not be an entire legal organization. It can denote a sub-team, product line, safety process, or other regulator-chosen unit.} At each
timestep $t \in [T]$, the auditor collects any new whistleblower
reports and selects an audit target. The collected reports form a
\emph{report stream}
$\rmR_T = (\rvr_1, \ldots, \rvr_T) \in \mathbb{N}_0^{C \times T}$,
where the $c$-th entry of $\rvr_t$, denoted $r_{t,c}$, counts the
number of new reports about audit target $c$ at time $t$. The
sequence of audit decisions forms the \emph{audit transcript}
$\rva_T = (a_1, \ldots, a_T) \in [C]^T$. For $t \le T$, we write
$\rmR_t$ and $\rva_t$ for the length-$t$ prefixes. We write $\gR_T := \mathbb{N}_0^{C \times T}$ for the space of report
streams and $\gA_T := [C]^T$ for the space of audit transcripts.


The auditing setup must satisfy the following conditions.

\begin{condition}[Privacy]
\label{cond:privacy}
The auditor should guarantee each report a fixed level of privacy, regardless of when it is submitted or which audit target it concerns.

\end{condition}

We formalize Condition~\ref{cond:privacy} using the following definitions of adjacency and privacy:

\begin{definition}[Adjacent Report Streams]
\label{def:adjacency}
Two report streams $\rmR_T, \rmR'_T$ are \emph{adjacent}, written
$\rmR_T \sim \rmR'_T$, if there exists a unique $(t^*,c^*) \in [T]\times[C]$
such that $|r_{t^*,c^*} - r'_{t^*,c^*}| = 1$ and $r_{t,c} = r'_{t,c}$ for all
$(t,c) \neq (t^*,c^*)$.
\end{definition}

\begin{definition}[$(\epsilon,\delta)$-Differential Privacy]
\label{def:DP}

A randomized mechanism $\gM : \gR_T \to \gA_T$ satisfies
$(\epsilon,\delta)$-\emph{differential privacy} with respect to the
adjacency relation of Definition~\ref{def:adjacency} if for all
adjacent $\rmR_T \sim \rmR'_T$ and all $A \subseteq \gA_T$,
\[
  \Pr[\gM(\rmR_T) \in A] \;\le\; e^{\epsilon} \cdot \Pr[\gM(\rmR'_T) \in A] + \delta.
\]
\end{definition}

For the special case of $(0,\delta)$-DP, we can demand equivalently: $d_{\mathrm{TV}}(\gM(\rmR_T),\gM(\rmR_T')) \le \delta$ for every adjacent pair (a standard fact proven in Proposition~\ref{app:prop:tv-char}).

\begin{remark}[Per-report and per-whistleblower privacy]
\label{rem:per-report-privacy}
Under definition~\ref{def:adjacency}, a privacy guarantee stated as in definition~\ref{def:DP} should be interpreted as per-report (\emph{event-level}). If a whistleblower files $k$ reports, then the inclusion of that whistleblower's full contribution may change $k$ report events. By standard group privacy, a mechanism satisfying $(\epsilon,\delta)$-DP for single-report adjacency satisfies \[ \left(k\epsilon,\; \delta \sum_{i=0}^{k-1} e^{i\epsilon}\right)\text{-DP} \] for the inclusion or exclusion of all $k$ reports by that whistleblower (see Appendix Lemma~\ref{app:lem:group_privacy_reports}). In the special case of $\epsilon=0$ used throughout this paper, this becomes $(0,k\delta)$-DP. Consequently, a nontrivial per-whistleblower guarantee requires a bound on the number of reports attributed to one whistleblower. 
\end{remark}

Audits based on whistleblower reports are often costly \citep{kuang2021whistleblowing} so it is reasonable for an auditor to audit less frequently than whistleblowers can issue reports, motivating the following practical condition:



\begin{condition}[Report Reset]
\label{cond:reset}
After a target is audited, previous reports about that target are no
longer counted toward future audit decisions.
\end{condition}

This condition reflects the interpretation of an audit as resolving the currently
pending concerns about that audit target. Thus, at time $t$, the relevant count for
target $c$ is not the cumulative number of reports over the entire horizon,
but the number of \emph{active} reports: reports received since the most recent
audit of $c$, or since $t=1$ if $c$ has not yet been audited.

\begin{definition}[Active counts]
\label{def:active-counts}
Fix a report stream prefix $\rmR_t$ and a decision prefix
$\rva_{t-1} \in [C]^{t-1}$. For each $c \in [C]$, let
$\tau_c(t) \coloneqq \max\{s \le t-1 : a_s = c\}$, with
$\max \emptyset \coloneqq 0$, denote the most recent audit of $c$
before time $t$. The \emph{active count} of target $c$ at time $t$ is
\[
    n_{t, c} \coloneqq \sum_{s = \tau_c(t)+1}^{t} r_{s, c},
\]
the number of reports about $c$ received since its last audit,
including those received at time $t$ itself. When the decisions are
produced by a mechanism $\gM$, the $n_{t, c}$ are random variables
through $\rva_{t-1}$.
\end{definition}

\section{A Threat Model}
\label{sec:threat_model}
\paragraph{The adversary.} Following the strong-adversary threat model of
\citet{kulynychunifying}, we take the natural adversary in our setting to be
the audit target itself: an entity with near-complete knowledge of its
internal state and reporting history. The adversary has an incentive to identify whistleblowers in
order to retaliate, knows the auditing mechanism $\gM$ including all parameters, and observes the realized audit transcript $\rva_T$. 

We instantiate the strong-adversary model as a binary adjacent-stream distinguishing game. Let $\rmR_T \sim \rmR'_T$ be two adjacent report streams that differ only in whether one report at $(t^*,c^*)$ was filed. The adversary's task is to distinguish the two worlds
\begin{align*}
H_1:\rva_T \sim \gM(\rmR_T)
\qquad
\text{vs.}
\qquad
H_0:\rva_T \sim \gM(\rmR'_T).
\end{align*}
This captures the relevant privacy question of whether observing the audit transcript lets the audit target infer that a particular report was filed.
Because the adversary is assumed to have near-complete background knowledge, bounding its distinguishing power also bounds the power of weaker adversaries. 


\paragraph{Tests, error rates, and the prior-free guarantee.}
A test for the distinguishing game above is a decision rule
$\varphi : \mathcal{A}_T \to [0,1]$, where $\varphi(\rva_T)$ is the
probability with which the adversary, having observed the transcript
$\rva_T$, rejects $H_0$ and concludes that the additional report was
filed; deterministic rules correspond to $\varphi(\rva_T) \in \{0,1\}$.
Its true- and false-positive rates,
\begin{align*}
\mathrm{TPR}(\varphi)
\coloneqq \mathbb{E}_{\rva_T \sim \gM(\rmR_T)}\bigl[\varphi(\rva_T)\bigr],
\qquad
\mathrm{FPR}(\varphi)
\coloneqq \mathbb{E}_{\rva_T \sim \gM(\rmR'_T)}\bigl[\varphi(\rva_T)\bigr],
\end{align*}
are the power and the size of the test $\varphi$. By the total-variation
characterization of $(0,\delta)$-DP
(Proposition~\ref{app:prop:tv-char}), the privacy guarantee is
equivalent to the following prior-free statement:
\begin{equation}
\mathrm{TPR}(\varphi) \;\le\; \mathrm{FPR}(\varphi) + \delta
\qquad
\text{for every adjacent pair $\rmR_T \sim \rmR'_T$ and every rule $\varphi$.}
\label{eq:prior-free}
\end{equation}
This becomes clear after plugging in the indicator function: $\varphi = \mathbbm{1}_A$, which recovers
Definition~\ref{def:DP} with $\epsilon = 0$, and conversely
$\mathrm{TPR}(\varphi) - \mathrm{FPR}(\varphi) \le
d_{\mathrm{TV}}\bigl(\gM(\rmR_T), \gM(\rmR'_T)\bigr)
\le \delta$ for every $[0,1]$-valued $\varphi$.
The corresponding interpretation thus is: no test on the transcript can detect the report at a rate
exceeding its false-alarm rate by more than $\delta$.
 
\paragraph{The advantage of the adversary.}
Following the baseline-aware perspective of \citet{kulynychunifying},
we translate~\eqref{eq:prior-free} into a statement about the
adversary's decisions by placing a weight on the two hypotheses. Fix
$\pi \in [0,1]$ and let $W \sim \mathrm{Bern}(\pi)$ select the true
hypothesis: the mechanism runs on $\rmR_T$ if $W = 1$ (hypothesis
$H_1$) and on $\rmR'_T$ if $W = 0$ (hypothesis $H_0$), and the
adversary --- who knows $\gM$, both streams, and $\pi$ ---
observes $\rva_T$ and guesses $W$. The \emph{success} of a rule
$\varphi$ is its probability of guessing correctly in this experiment;
by the law of total probability,
\begin{align*}
\mathrm{succ}_{\pi,\rmR_T,\rmR'_T}(\varphi;\gM)
\coloneqq
\Pr[\text{guess} = W]
=
\pi\,\mathrm{TPR}(\varphi)
+ (1-\pi)\bigl(1 - \mathrm{FPR}(\varphi)\bigr).
\end{align*}
The weight $\pi$ may be read either as the adversary's prior belief
that the report was filed, or as the frequency of the reporting event
across the scenarios the guarantee must cover; all bounds below hold
uniformly in $\pi$, so no assumption on $\pi$ is made. Without
observing $\rva_T$, the best achievable success is that of the better
constant rule ($\varphi \equiv 1$ or $\varphi \equiv 0$):
\begin{align*}
\mathrm{base}_{\pi} \coloneqq \max\{\pi, 1-\pi\}.
\end{align*}
The \emph{additive advantage} of $\varphi$ is its improvement over this
baseline,
\begin{align*}
\mathrm{adv}_{\pi,\rmR_T,\rmR'_T}(\varphi;\gM)
\coloneqq
\mathrm{succ}_{\pi,\rmR_T,\rmR'_T}(\varphi;\gM)
- \mathrm{base}_{\pi},
\end{align*}
so that any rule which ignores the transcript has advantage at most
zero. The mechanism's
advantage at weight $\pi$ is the largest such improvement over all
adjacent streams and decision rules:
\begin{align*}
\mathrm{adv}_{\pi}(\gM)
\coloneqq
\sup_{\rmR_T \sim \rmR'_T}\,
\sup_{\varphi}\,
\mathrm{adv}_{\pi,\rmR_T,\rmR'_T}(\varphi;\gM).
\end{align*}
 
\begin{theorem}[Prior-independent deniability]
\label{thm:advantage}
If $\gM$ satisfies $(0,\delta)$-differential privacy, then for
every $\pi \in [0,1]$,
\begin{align*}
\mathrm{adv}_{\pi}(\gM)
\;\le\;
\min\{\pi, 1-\pi\}\,\delta.
\end{align*}
In particular, the prior-independent advantage
$\mathrm{adv}(\gM)
\coloneqq \sup_{\pi \in [0,1]} \mathrm{adv}_{\pi}(\gM)$
satisfies $\mathrm{adv}(\gM) \le \delta/2$.
\end{theorem}
 
\begin{proof}
Fix an adjacent pair $\rmR_T \sim \rmR'_T$ and a rule $\varphi$, and
abbreviate $\mathrm{TPR} = \mathrm{TPR}(\varphi)$,
$\mathrm{FPR} = \mathrm{FPR}(\varphi)$. If $\pi \le 1/2$, then
$\mathrm{base}_{\pi} = 1 - \pi$ and
\begin{align*}
\mathrm{adv}_{\pi,\rmR_T,\rmR'_T}(\varphi;\gM)
= \pi\,\mathrm{TPR} - (1-\pi)\,\mathrm{FPR}
\le \pi\,(\mathrm{TPR} - \mathrm{FPR})
\le \pi\,\delta,
\end{align*}
where the first inequality uses $\mathrm{FPR} \ge 0$ and
$1 - \pi \ge \pi$, and the second is Eq.~\eqref{eq:prior-free}. The
case $\pi \ge 1/2$ is symmetric:
\begin{align*}
\mathrm{adv}_{\pi,\rmR_T,\rmR'_T}(\varphi;\gM)
= (1-\pi)(1 - \mathrm{FPR}) - \pi(1 - \mathrm{TPR})
\le (1-\pi)(\mathrm{TPR} - \mathrm{FPR})
\le (1-\pi)\,\delta,
\end{align*}
using $1 - \mathrm{TPR} \ge 0$ and $\pi \ge 1 - \pi$. Taking suprema
over $\varphi$ and over adjacent pairs, and noting
$\min\{\pi, 1-\pi\} \le 1/2$, gives both claims.
\end{proof}
 
We adopt $(0,\delta)$-DP as the operational privacy definition for the
rest of the paper. By Theorem~\ref{thm:advantage}, no strong adversary
--- including the audit target itself --- can improve its probability
of correctly determining whether a single report was filed by more
than $\delta/2$, whichever prior it holds.

\begin{remark}[From reports to whistleblowers]
\label{rem:per-whistleblower}
By group privacy, if a whistleblower files at most $k$ reports, the transcript
satisfies $(0,k\delta)$-DP with respect to that whistleblower's full
contribution (cf. Appendix Corollary~\ref{app:cor:per_whistleblower}).
Thus Theorem~\ref{thm:advantage} gives person-level additive advantage at most
$k\delta/2$. An auditor targeting person-level advantage
$\alpha_{\mathrm{wb}}$ under a repeat-reporting budget $k_{\max}$ can set
$\delta=2\alpha_{\mathrm{wb}}/k_{\max}$.
\end{remark}

\begin{remark}[Intermediate observation times]
\label{rem:prefix}
The guarantee holds at every intermediate time simultaneously. Any prefix $\rva_{t}$ is itself $(0,\delta)$-DP, and
any decision rule on $\rva_{t}$ induces a transcript rule with
identical error rates, so the advantage bound of
Theorem~\ref{thm:advantage} holds verbatim at every $t \le T$; see Appendix~\ref{app:rem:prefix} for details.
\end{remark}

\paragraph{Why this threat model is different.}
Our setting differs from prior work in two ways. First, the privacy unit is a
single report, but the public observable is the entire sequence of audit
decisions in $[C]$, rather than a noisy count vector or histogram. Second,
prior whistleblower mechanisms
\citep{warner1965randomized,chassang2019crime} randomize each report at
submission time under an untrusted-auditor model. We instead trust the auditor
with the reports but treat the audit transcript as public and observed by a
powerful insider adversary. Thus privacy must hold for a report's effect on the
full audit transcript over horizon $T$, not merely for a privatized message.

\paragraph{Utility.} If privacy were not a concern, we assume the auditor
would like to select the audit target with the largest number of
whistleblower reports to audit next. To quantify this, we define the error of
an auditing mechanism $\gM$ as the probability it does not
select the audit target with the current largest number of reports.



\begin{definition}[Error]
    \label{def:error}
    Fix a time $t \in [T]$ and a report stream prefix $\rmR_t$. Let
    \[
        C_t^* \coloneqq \argmax_{c \in [C]} n_{t, c}
    \]
    be the set of audit targets with the largest active count
    (Definition~\ref{def:active-counts}) immediately before the audit
    decision at time $t$. The error of a mechanism $\gM$ at time $t$
    is the mis-selection probability
    \[
        \Pr[a_t \notin C_t^*],
    \]
    where the probability is over the mechanism's randomness, and both
    $a_t$ and $C_t^*$ are random through the mechanism's own decision
    prefix $\rva_{t-1}$.
\end{definition}

We now have a concrete way to compare mechanisms. For any mechanism $\gM$
that satisfies Conditions~\ref{cond:privacy} and \ref{cond:reset}, the
auditor offers a fixed per-report DP parameter $\delta$, which implies a
baseline-independent additive-advantage bound of $\delta/2$
(Theorem~\ref{thm:advantage}). For this privacy level, a better mechanism has
strictly lower error, as defined in Definition~\ref{def:error}. The next
natural questions are what guarantees exist for existing mechanisms, and
whether we can do better.



\section{Randomized Response as a Natural Baseline}
\label{sec:rand_response}

A natural baseline is to apply \emph{randomized response} at the audit-selection step~\citep{warner1965randomized}: with some probability $p_{\mathrm{rand}}$ the auditor ignores the reports and selects an audit target uniformly at random, and otherwise selects the target with the largest current report count. Prior work has applied this to whistleblowing~\citep{chassang2019crime}. The mechanism is described in Algorithm~\ref{app:alg:randomized_response}.

Randomized response satisfies
$(0,\, 1 - p_{\mathrm{rand}}^{\,T})$-DP at horizon $T$
(Proposition~\ref{app:prop:RR_privacy}), and its one-step error is at
most $\frac{C-1}{C}\,p_{\mathrm{rand}}$, with equality whenever the
current maximizer is unique (Proposition~\ref{app:prop:RR_error}).
Calibrating through Proposition~\ref{app:prop:RR_privacy} requires
$p_{\mathrm{rand}} \ge (1-\delta)^{1/T}$, which tends to $1$ as the
horizon grows; this is the calibration we use for the experimental
baseline in Section~\ref{sec:experiments}. 
However, a smaller exploration probability also cannot reconcile a fixed privacy level with informative audit selection: the privacy requirement pins randomized response to within $\delta$ of
uniform random auditing, already at the first audit decision. The following theorem shows this:
 
\begin{theorem}[Privacy pins randomized response to near-uniform auditing]
\label{thm:RR_tradeoff}
Fix $C \ge 2$, a horizon $T \ge 1$, and $p_{\mathrm{rand}} \in [0,1]$.
If Algorithm~\ref{app:alg:randomized_response} satisfies
$(0,\delta)$-DP at horizon $T$ for some $\delta \in [0,1)$, then
\[
    p_{\mathrm{rand}} \;\ge\; 1 - \frac{\delta C}{C-1}.
\]
Consequently, for every $t \le T$ and every pre-decision history
$\mathcal{H}_{t-1} = h$ under which the active-count maximizer $c_t^*$
is unique,
\[
    \Pr[a_t \ne c_t^* \mid \mathcal{H}_{t-1} = h]
    \;=\; \frac{C-1}{C}\,p_{\mathrm{rand}}
    \;\ge\; \frac{C-1}{C} - \delta.
\]
In particular, no calibration of randomized response, at any horizon,
improves on the error $\frac{C-1}{C}$ of uniform random auditing by
more than $\delta$. The full proof can be found in section~\ref{app:subsubsec:proof_rr}.
\end{theorem}
 
\begin{remark}[Tightness, and no degradation over time]
\label{rem:RR_tightness}
The requirement of Theorem~\ref{thm:RR_tradeoff} is tight up to the
constant multiplying $\delta$: the choice
$p_{\mathrm{rand}} = C/(C+\delta)$ already suffices for
$(0,\delta)$-DP at \emph{every} horizon simultaneously
(Proposition~\ref{app:prop:RR_horizon_free}), because the
report reset (Condition~\ref{cond:reset}) erases the influence of a
report the first time its target is audited, so leakage does not
accumulate with $T$. Hence the optimal error of $(0,\delta)$-DP
randomized response lies in
$\bigl[\tfrac{C-1}{C} - \delta,\; \tfrac{C-1}{C+\delta}\bigr]$: a
gap-insensitive constant within $\Theta(\delta)$ of uniform auditing,
independent of the horizon. The failure of randomized response is thus
not that its privacy degrades over time, but that at any useful
privacy level it must explore almost always and therefore cannot react
to the report gap --- in contrast to the mechanism of
Section~\ref{sec:generic}, whose error vanishes as the gap
grows.
\end{remark}

\section{A Generic Reduction to Private Continual Counting}
\label{sec:mechanism}

Randomized response solves the privacy problem at the selection step; however, by Theorem~\ref{thm:RR_tradeoff} this pins its error within $\delta$ of uniform random auditing, however large the report gap. Using \emph{continual counting}
mechanisms, we instead privatize the \emph{report counts}, so that a growing
gap can outrun the noise. These mechanisms take a stream of non-negative
counts and output noisy running totals at each step:

\begin{definition}[Continual-counting mechanism]
\label{def:ccm}
Fix a horizon $T \in \mathbb N$. A \emph{continual-counting mechanism}
$\mathcal C$ is a randomized online process which, when fed a stream
$x_1,\dots,x_T \in \mathbb N_0$ sequentially, outputs after each step
$t \in [T]$ a noisy estimate $\hat s_t$ of the prefix sum
$s_t := \sum_{i=1}^{t} x_i$.
Equivalently, for each $x \in \mathbb N_0^T$, the mechanism induces a random
output vector $\mathcal C(\rvx) = (\hat s_1,\dots,\hat s_T) \in \mathbb{R}^T$.
\end{definition}

\subsection{A Generic Private Auditing Algorithm}
\label{sec:generic}

We now give a generic reduction from private continual counting to private
auditing (Algorithm~\ref{alg:mechanism}). Given any continual-counting mechanism $\gC$, the auditing mechanism
$\calM_{\gC}$ maintains, for each audit target, an independent noisy counter for
the number of reports received since that target was last audited. 
At each time step, the auditor observes the noisy pending counts and audits a
target with maximum noisy count. The mechanism implements Condition~\ref{cond:reset} (Report Reset) by restarting an audit target's counter after that target is audited.

The privacy argument then reduces to the privacy of the affected continual
counter, with target selection and the audit transcript obtained from the noisy
counter outputs.

\begin{algorithm}[H]
\caption{Generic Private Auditing Mechanism $\gM_\gC$}
\label{alg:mechanism}
\begin{algorithmic}[1]
\Require Number of audit targets $C$; horizon $T$;
         $(0,\delta)$-DP continual-counting mechanism $\gC^{(T)}$
\State \textbf{Initialize.} For each $c \in [C]$,
  start a fresh instance $\gC^{(T)}_c$ of $\gC^{(T)}$, and set $\ell_c \gets 0$.
\For{$t = 1, 2, \ldots, T$}
  \For{each $c \in [C]$}
  \State \textbf{Receive} new report counts $r_{t,c} \in \N_0$
    \State $\ell_c \gets \ell_c + 1$ \quad (time since last audit of $c$)
    \State feed $r_{t,c}$ to $\gC^{(T)}_c$ as its $\ell_c$-th input;
           let $\tilde n_{t,c}$ be the $\ell_c$-th output
           \label{line:count}
  \EndFor
  \State \textbf{Select} $a_t \gets \arg\max_{c \in [C]}\, \tilde{n}_{t,c}$
    \quad (ties broken uniformly at random). \label{line:select}
  \State \textbf{Restart:} after auditing $a_t$, replace $\gC^{(T)}_{a_t}$ with a fresh instance of $\gC^{(T)}$ and set $\ell_{a_t} \gets 0$.
    \label{line:restart}
\EndFor
\end{algorithmic}
\end{algorithm}

\subsection{Privacy Guarantee}
\label{sec:privacy}


\begin{proposition}[Privacy]\label{prop:privacy_reduction}
Fix a horizon $T \in \mathbb{N}$. If $\gC$ satisfies $(0,\delta)$-differential privacy as a continual-counting mechanism on streams of length $T$, then Algorithm~\ref{alg:mechanism}, instantiated with $\gC$, satisfies $(0,\delta)$-differential privacy with respect to adjacent report streams $\rmR_T \sim \rmR_T'$. That is, for all measurable $A \subseteq \gA$,
\vspace{-0.25em}
\begin{align*}
  \Pr[\gM_\gC(\rmR_T) \in A]
  \le
  \Pr[\gM_\gC(\rmR_T') \in A] + \delta .
  \label{eq:privacy_reduction}
\end{align*}
\vspace{-0.75em}
\end{proposition}

\begin{proof}[Proof sketch]
Let $\rmR_T \sim \rmR_T'$ differ by one report for audit target $c^\ast$
at time $t^\ast$. Before $t^\ast$, the two executions have identical
counter inputs and hence identical restart histories. The extra report can
therefore enter only one active counter instance, whose two input streams are
adjacent. By the assumed $(0,\delta)$-DP guarantee for the continual counter,
the outputs of this affected instance differ by at most $\delta$ in total
variation; all other counter randomness can be coupled identically. The audit
transcript is a post-processing of these counter outputs and restart decisions,
so it inherits the same $(0,\delta)$-DP guarantee. See
Appendix~\ref{app:privacy_reduction_proof} for the full proof.
\end{proof}

\begin{remark}[Per-report privacy guarantee]\label{rem:per-report}

  Proposition~\ref{prop:privacy_reduction} guarantees that the
  audit transcript leaks at most $\delta$ additional information about whether any
  \emph{single specific report} was filed.  If a whistleblower files $k$ reports
  (possibly at different times or about different organizations), their overall
  privacy follows by basic composition \citep{dwork2014algorithmic} and is at
  most $(0, k\delta)$-DP.  
\end{remark}

\section{Concrete Continual-Counting Instantiation}
\label{sec:counter-instantiation}
We now construct a continual-counting mechanism calibrated for $(0,\delta)$-DP.

\subsection{Gaussian Matrix-Mechanism Counter}
\label{subsec:gaussian-counter}

Let $\rmA^{(T)}\in\{0,1\}^{T\times T}$ denote the lower-triangular all-ones prefix-sum matrix, i.e. $\rmA^{(T)}_{i,j} = \mathbf 1[i\ge j]$.
Thus the exact vector of prefix sums of $\rvx\in\mathbb N_0^T$ is
$\rmA^{(T)}\rvx$. 
Suppose we are given matrices $\rmB^{(T)},\rmC^{(T)}\in\mathbb R^{T\times T}$ such that $\rmB^{(T)}\rmC^{(T)} = \rmA^{(T)}$.
We define the associated \emph{Gaussian matrix-mechanism counter} by
\begin{equation}
\label{eq:generic-counter}
\gC(\rvx)
:=
\rmB^{(T)}\!\bigl(\rmC^{(T)}\rvx +\rvz\bigr),
\;
\rvz \sim \mathcal N(0,\sigma^2 \rmI_T).
\end{equation}

For privacy, the key quantity is the largest column norm of the encoder:
\begin{equation}
\label{eq:definition_mt}
M_T \;:=\; \max_{j\in[T]} \|\rmC^{(T)}\rve_j\|_2
\end{equation}

\begin{proposition}[Base-counter Privacy]
\label{prop:base-counter-privacy}
Fix $T\in\mathbb N$.  Suppose $\gC$ is defined by
eq.~(\ref{eq:generic-counter}), with $\rmB^{(T)}\rmC^{(T)}=\rmA^{(T)}$. If the noise scale is chosen as

\begin{equation}
\label{eq:def_variance}
\sigma
=
\frac{M_T}{2\,\Phi^{-1}((1+\delta)/2)},
\end{equation}
where $M_T$ is defined as in eq.~(\ref{eq:definition_mt}),
then $\gC$ is $(0,\delta)$-differentially private. Throughout, we write $\kappa_\delta := \Phi^{-1}\!\left((1+\delta)/2\right)$, so that the noise scale can equivalently be written as $\sigma = M_T / (2\kappa_\delta)$.
\end{proposition}

The proof (Appendix~\ref{app:base_counter_privacy_proof}) follows by applying Lemma~\ref{app:lem:tv-shifted} to the encoded shift $\rmC^{(T)}\rve_{j^*}$, whose norm is at most $M_T$, and post-processing through $\rmB^{(T)}$.

\subsection{Concrete Instantiation via Toeplitz Factorization}
\label{subsec:tca-factorization}

To maximize utility, we use the Toeplitz factorization of \citet{fichtenberger2023constant} and \citet{dvijotham2024efficient}, which is near-optimal for continual counting.

We define the $T \times T$ lower-triangular Toeplitz
matrices $\rmC^{(T)}$ (encoder) and $\rmB^{(T)}$ (decoder) by
\begin{equation}
\label{eq:def_toeplitz_matrices}
  \rmC^{(T)}_{i,j} = \rmB^{(T)}_{i,j} :=
  \begin{cases} f_{i-j} & \text{if } i \ge j, \\ 0 & \text{if } i < j. \end{cases}
\end{equation}
where the sequence $(f_k)_{k \ge 0}$ is given by
\begin{equation}
      f_0 = 1, \qquad f_k = \Bigl(1 - \frac{1}{2k}\Bigr)f_{k-1}
  \quad\text{for } k \ge 1
\end{equation} 

The convolution properties of this sequence yields $\rmB^{(T)}\rmC^{(T)}=\rmA^{(T)}$ (Lemma 1 \& 2 \cite{fichtenberger2023constant}).

\begin{lemma}[Sensitivity Bound for the Toeplitz Encoder]
\label{lem:tca-sensitivity}
For the encoder $\rmC^{(T)}$ above,
\[
\|\rmC^{(T)}\rve_j\|_2^2 = \sum_{k=0}^{T-j} f_k^2
\le
\sum_{k=0}^{T-1} f_k^2
\]
which equals $M_T^2$ eq.~(\ref{eq:definition_mt}), as the maximum is attained at $j=1$.
Moreover, $M_T = \mathcal{O}(\sqrt{\log T})$.
\end{lemma}

\begin{proof}
The first identity is immediate from the definition of $\rmC^{(T)}$.  The
$\mathcal{O}(\sqrt{\log T})$ bound follows from the analysis in
\citet{dvijotham2024efficient} (Lemma 2.1).
\end{proof}

Instantiating Algorithm~\ref{alg:mechanism} with this Toeplitz-based counter -- running one independent copy per organization and restarting after each audit -- yields our final mechanism, which satisfies $(0,\delta)$-DP. This factorization was shown to be optimal among lower-triangular Toeplitz factorizations for continual counting~\cite{dvijotham2024efficient}. We use this instantiation throughout the remainder of the paper.

\section{Utility Guarantee}
\label{sec:utility}
We now derive a utility guarantee (Definition~\ref{def:error}) for the above mechanism. To begin, define the noise history of the mechanism.

\begin{definition}[History]\label{def:history}
$\calH_{t-1} := \sigma\bigl(\{\tilde{n}_{c,s} : c \in [C],\, s \le t-1\}\bigr)$ is the $\sigma$-algebra generated by all noisy counter outputs up to time $t-1$. The audit decisions $a_1,\dots,a_{t-1}$ are $\calH_{t-1}$-measurable.
\end{definition}

\subsection{Utility Theorem}
\label{sec:utility-thm}

\begin{definition}[Effective gap]\label{def:eff-gap}
For any realisation of $\calH_{t-1}$ and any challenger $c \ne c^*_t$, the
\emph{effective gap} between true leader $c^*_t$ and challenger $c$ at time $t$
is
\begin{equation}\label{eq:eff-gap}
  \tilde{\Delta}_{t, c}
  \;:=\;
  \underbrace{(n_{t, c^*_t} - n_{t, c})}_{\text{true gap}\ \Delta_{t, c}}
  \;+\;
  \underbrace{(G_{t, c^*_t} - G_{t, c})}_{\text{accumulated past noise difference}}
\end{equation}
where $\Delta_{t, c} = n_{t, c^*_t} - n_{t, c} > 0$ and $G_{t, c}$ is defined
in eq.~(\ref{eq:Gct}).  Note that $\tilde{\Delta}_{t, c}$ is fully determined by
$\calH_{t-1}$ and the data.
\end{definition}

\begin{theorem}[Utility of Algorithm~\ref{alg:mechanism} with Toeplitz Factorization]
\label{thm:utility}
Fix any time $t \ge 1$.  For any realisation $h$ of $\calH_{t-1}$ such that
$c^*_t$ is the \emph{unique} true leader ($\Delta_{t, c} > 0$ for all
$c \ne c^*_t$), the conditional error satisfies
\begin{equation}\label{eq:cond-utility}
  \Pr\bigl[a_t \ne c^*_t \mid \calH_{t-1} = h\bigr]
  \;\le\;
  \sum_{c\,\ne\,c^*_t}
  \Phi\!\left(-\frac{\sqrt{2}\,\kappa_\delta\,\tilde{\Delta}_{t, c}}{M_T}\right),
\end{equation}
where each term is \emph{exact} (equality, not merely an upper bound) for its
corresponding challenger.  

\end{theorem}

For a proof, see Appendix~\ref{app:subsec:utility}.

\subsection{Remarks}
\label{sec:utility-remarks}

\begin{remark}[Optimality]
\label{rem:tca-opt}
The error bound eq.~(\ref{eq:cond-utility}) depends on the noise calibration only
through $\sigma = M_T/(2\kappa_\delta)$. Since smaller $M_T$ means smaller
$\sigma$ and hence smaller error, and since Proposition~2.2 of
\citet{dvijotham2024efficient} shows the above factorization uniquely minimises
$M_T = \sqrt{\mathrm{MaxErr}(\rmB^{(T)},\rmC^{(T)})}$ over all lower-triangular
Toeplitz factorisations, the above mechanism provides the tightest error bound within this
class at any fixed privacy level $\delta$.
\end{remark}

 


\begin{remark}[Comparison with Randomized Response]
\label{rem:rr-comparison}
Theorem~\ref{thm:RR_tradeoff} shows that any $(0,\delta)$-DP calibration of
randomized response has error at least $\frac{C-1}{C} - \delta$ at every
timestep, independent of the report gap. By contrast, the bound of
Theorem~\ref{thm:utility} vanishes whenever the effective gap satisfies
$\tilde{\Delta}_{t,c}/\sqrt{\log T} \to \infty$, as when gaps grow with
accumulated reports. At any fixed privacy level, randomized response is thus
gap-insensitive while our mechanism's error decays as the gap grows.
\end{remark}


\section{Experiments}
\label{sec:experiments}
We compare \emph{Toeplitz Continual Auditing} (TCA), with Randomized Response \cite{warner1965randomized}, uniformly random auditing, and non-private (greedy) auditing in two simulation studies. All private mechanisms are calibrated to the same per-report $(0, \delta)$-DP guarantee over the relevant horizon; full details can be found in Appendix~\ref{app:sec:exp_details}.

\paragraph{Experiment 1: Static gap sweep.}
Theorem~\ref{thm:utility} predicts that TCA's error should decrease as the effective gap between the true leader and its challengers grows, with noise scale governed by $M_T=\mathcal{O}(\sqrt{\log T})$. To isolate this effect, we consider a single audit decision with one leading organization $c_0$ having active count $n_0=\Delta$ and all other organizations having active count $n_c=0$. We vary $\Delta$ and estimate the mis-selection probability $\Pr[a \neq c_0]$. Figure~\ref{fig:exp1_misselection_vs_gap_c5} shows the expected qualitative pattern: TCA's error decreases rapidly as the gap grows, while randomized response is essentially insensitive to the gap. The advantage increases with the number of audited organizations (Appendix Figures~\ref{fig:misselection_vs_gap_c20}--\ref{fig:misselection_vs_gap_c200}).

\begin{figure*}[t]
    \centering

    \begin{subfigure}[t]{0.48\linewidth}
        \centering
        \includegraphics[width=\linewidth]{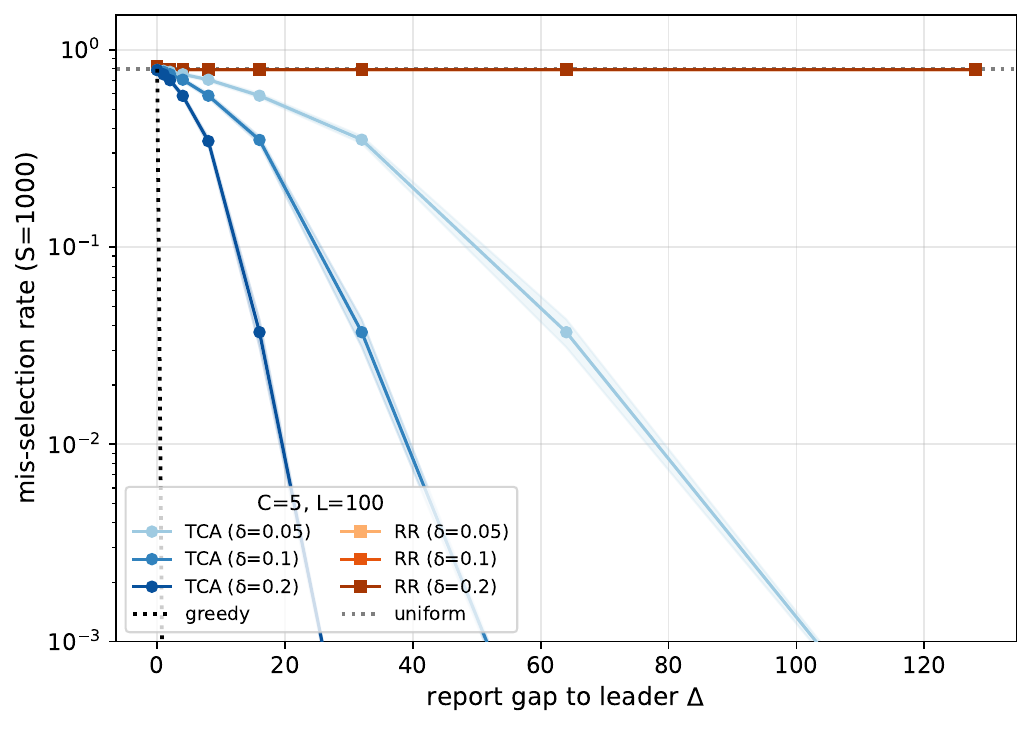}
        \caption{Mis-selection rate as a function of the leading gap.}
        \label{fig:exp1_misselection_vs_gap_c5}
    \end{subfigure}
    \hfill
    \begin{subfigure}[t]{0.51\linewidth}
        \centering
        \includegraphics[width=\linewidth]{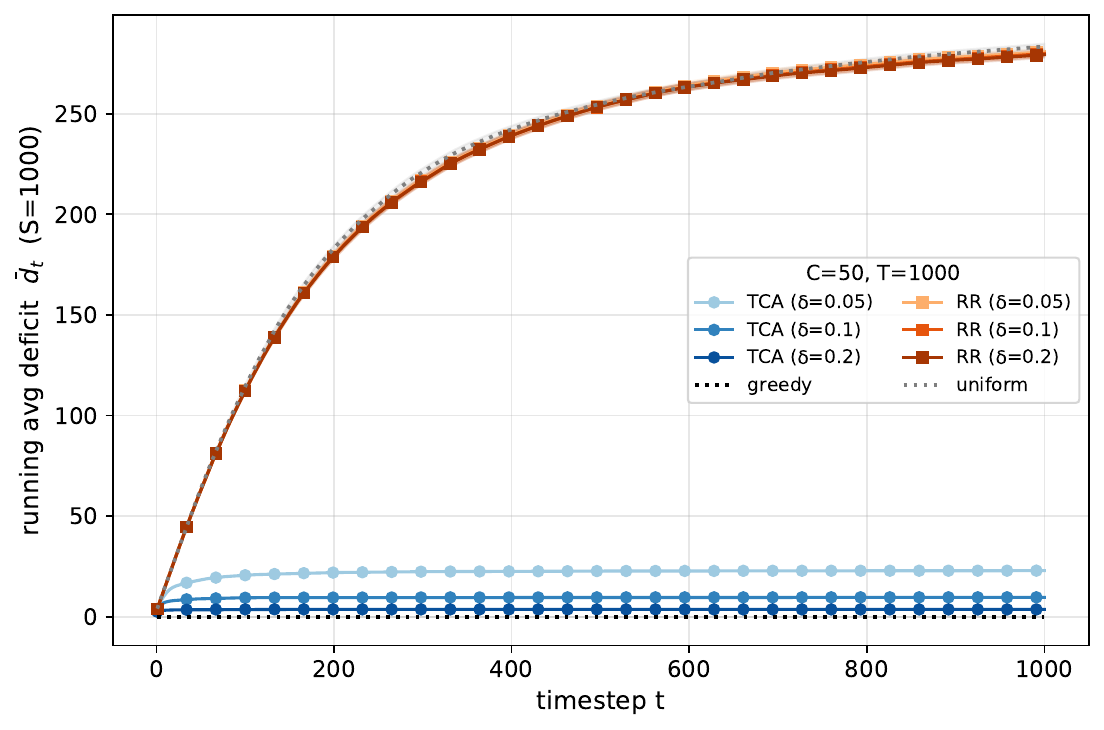}
        \caption{Running average active-count deficit over time.}
        \label{fig:exp2_running_avg_deficit}
    \end{subfigure}

    \caption{\textbf{Utility Simulation.}
    (a) Static gap sweep: TCA's mis-selection rate decreases as the
leader--runner-up gap $\Delta$ grows, while randomized response remains
nearly gap-insensitive. 
(b) Dynamic auditing: TCA incurs substantially lower running active-count
deficit than randomized response and uniform auditing, approaching the
non-private greedy baseline.}
    \label{fig:experiments}
\end{figure*}

\textbf{Experiment 2: Dynamic auditing.} The static experiment isolates a single decision, but real auditing is online: each audit resets the selected organization's active count, so different mechanisms induce different future states. We therefore simulate a streaming report process over a fixed horizon $T$ and measure the \emph{active count deficit}: $d_t^{\mathcal M} := \max_{c} n^{\mathcal M}_{t,c} - n^{\mathcal M}_{t, a_t}$, computed before the reset at time $t$ — the regret of $\gM$'s selection against the best available active count. Figure~\ref{fig:exp2_running_avg_deficit} shows that TCA remains close to the greedy auditor, while randomized response accumulates much larger deficits. 

\section{Limitations}
\label{sec:limitations}
First, the per-whistleblower guarantee degrades linearly in the number of
reports $k$ under group privacy
(Lemma~\ref{app:lem:group_privacy_reports}, Corollary~\ref{app:cor:per_whistleblower}).
Second, the threat model assumes that the auditor is a trusted entity that will not leak the private report counts. While this is a standard threat model assumption in centralized differential privacy, it may be worth strengthening the threat model by omitting a trusted auditor. Mechanisms that satisfy local differential privacy \citep{kasiviswanathan2011can} such as the randomized response approach described in Section~\ref{sec:rand_response} would likely be sufficient to guarantee privacy under this threat model. 
Third, we were unable to obtain any data from a real-world auditing setup. An evaluation with real-world data could expose additional practical considerations that may have been missed in this work.

\section{Related Work}
\label{sec:related}
\paragraph{Whistleblower mechanisms in economics and computer science.}
The economics literature has long studied the design of intervention rules that protect informants from retaliation \citep{chassang2019crime,ortner2018making}, with randomized response \citep{warner1965randomized} supplying the canonical garbling primitive. Empirical auditing work confirms that whistleblower reports causally affect downstream audit intensity \citep{kuang2021whistleblowing}. On the systems side, deployed pipelines such as SecureDrop and academic designs Dissent and Riposte \citep{10.1109/SP.2015.27} target sender-anonymity using mix-nets, DC-nets, or PIR. Recent cryptographic work establishes that anonymous transfer without trusted parties is fundamentally hard \citep{10.1007/978-3-031-22365-5_24, 10.1007/978-3-031-48621-0_1}, motivating our information-theoretic, DP-based approach in which deniability is provided by the auditor's randomized policy rather than by network-level anonymity.

\paragraph{Differential privacy under continual observation.}
Our mechanism builds on the model introduced by \citet{dwork2010differential} and refined by \citet{chan2011private}, which gives polylogarithmic-error counters and top-$k$ extensions. Recent matrix-factorization mechanisms achieve near-optimal constants \citep{fichtenberger2023constant,henzinger2023almost} and admit efficient streaming implementations \citep{dvijotham2024efficient, andersson2025streamingprivatecontinualcounting}. We instantiate our generic counter-and-reset scheme with a Toeplitz factorization to obtain $O(\sqrt{\log T})$ noise.

\paragraph{Private selection and online decision making.} At each audit step we solve a private argmax over noisy counters, a problem whose canonical solutions are the exponential mechanism \cite{10.1109/FOCS.2007.41}, Report-Noisy-Max \cite{dwork2014algorithmic}, and Permute-and-Flip \cite{NEURIPS2020_01e00f2f}. When candidate scores are themselves DP outputs, \citet{10.1145/3313276.3316377} and \citet{cohen2024lower} give general selection frameworks. Closest in structure to our work are differentially private multi-armed bandit algorithms that maintain per-arm tree-counters \cite{Tossou_Dimitrakakis_2016}. However, they differ on three concrete axes: 1.~ their adjacency unit is a single pull or reward (vs. a single report on our transcript); 2.~ their objective is cumulative regret over an exploration horizon (vs. count-leader identification at every step under a fixed $(0, \delta)$ constraint); and 3.~ their feedback model uses stochastic rewards after each pull to update arm estimates (vs. direct report counts with no reward channel from the audit). Bandit regret bounds therefore do not transfer to deniability guarantees, and our utility theorem (Theorem~\ref{thm:utility}) has no analogue in the bandit literature.

\paragraph{Privacy semantics and AI governance auditing.}
Translating $(\epsilon,\delta)$-DP into operational deniability guarantees has recently been unified across re-identification, attribute-inference, and reconstruction risks via $f$-DP \citep{kulynychunifying}; this is also the framework that links DP to GDPR-style ``singling out'' \citep{cohen2020towards}. Empirical auditing of DP claims complements this view \citep{NEURIPS2023_9a6f6e0d}. In AI governance, calls for external scrutiny of frontier systems \citep{anderljung2023frontierairegulationmanaging, Kolt_Anderljung_Barnhart_Brass_Esvelt_Hadfield_Heim_Rodriguez_Sandbrink_Woodside_2024}, responsible reporting frameworks \citep{Kolt_Anderljung_Barnhart_Brass_Esvelt_Hadfield_Heim_Rodriguez_Sandbrink_Woodside_2024}, and end-to-end algorithmic-audit processes \citep{rajiclosing2020} all envision insider information as a key signal. Our work supplies the missing technical primitive: an auditor policy whose randomization gives whistleblowers a quantifiable deniability guarantee even after many rounds of audits---he setting in which randomized response is provably pinned within $\delta$ of uniform auditing at any horizon.

\section{Conclusion}
\label{sec:conclusion}

We formalized an auditing setup for whistleblower anonymity and showed that
the canonical primitive, randomized response, cannot improve on uniform
random auditing by more than $\delta$ at any fixed $(0,\delta)$-level, at
any horizon. We then proposed a generic counter-and-restart mechanism whose
privacy is inherited by post-processing from any $(0,\delta)$-DP
continual-counting mechanism, and instantiated it with the Toeplitz
factorization of \citet{fichtenberger2023constant} and
\citet{dvijotham2024efficient}, yielding per-report $(0,\delta)$-DP with
noise calibrated by $M_T = \mathcal{O}(\sqrt{\log T})$. Our utility theorem
shows that the error decays whenever the effective leader--challenger gap
grows faster than $\sqrt{\log T}$. Natural extensions include tightening
the per-whistleblower guarantee for multiple reports beyond group privacy,
and empirical evaluation on real audit pipelines. While this mechanism is
intended to make reporting channels more credible by limiting what audit
transcripts reveal about individual reports, the guarantee applies only
under the stated threat model and should not be taken to protect against
operational failures such as access-log leaks, non-anonymous reporting
channels, overly narrow audit scopes, or retaliation based on side
information.

\subsection*{Acknowledgments}
This research was enabled by support provided
by Mila -- Québec AI Institute and by École Polytechnique de Montréal. LR acknowledges support from the Engineering and Physical Sciences Research Council
with grant number EP/S021566/1. MJK is supported through the Canada CIFAR AI Chair program and by a grant from the IVADO R10 cluster.

\bibliographystyle{abbrvnat}
\bibliography{refs}

\clearpage

\appendix

\renewcommand{\contentsname}{Contents of Appendix}
\addtocontents{toc}{\protect\setcounter{tocdepth}{2}}
{
 \hypersetup{hidelinks}
 \tableofcontents
}

\newpage


\appendix

\onecolumn

\section{Threat Model Notation}
\label{app:sec:threat_model_notation}


Table~\ref{tab:object-mapping} summarizes how the record-level DP objects of
\citet{kulynychunifying} instantiate in our whistleblower-auditing setting;
see Remark~\ref{rem:kulynych-relation} for the formal correspondence.

\begin{table}[t]
\centering
\caption{Mapping between the record-level DP objects of \citet{kulynychunifying}
and the whistleblower-auditing objects of this paper. The strong adversary's knowledge
of the partial dataset $\bar S$ is encoded in our setting by the adjacency relation,
which fixes everything except the single additional report.}
\label{tab:object-mapping}
\renewcommand{\arraystretch}{1.3}
\begin{tabular}{@{}p{0.46\linewidth} p{0.46\linewidth}@{}}
\toprule
\textbf{Kulynych et al. (record-level DP)} & \textbf{This paper (whistleblower auditing)} \\
\midrule
Dataset $S$
  & Report stream $\rmR_T$ \\
Mechanism $M(S) \to \theta \in \Theta$
  & Mechanism $\gM(\rmR_T) \to \rva_T \in [C]^T$ \\
Privacy unit: one record $z$
  & One report at coordinate $(t^*, c^*)$ \\
Neighbours $S \simeq S'$ (add/remove $z$)
  & Adjacent streams $\rmR_T \sim \rmR_T'$ (one report differs) \\
Strong adversary knows partial dataset $\bar S$
  & Audit target knows both $\rmR_T$ and $\rmR_T'$ (all but the additional report) \\
Prior $z \sim P$; Bernoulli case $z = z_b$, $b \sim \mathrm{Bern}(\pi)$
  & Prior $\pi$ over \{filed, silent\} (bounds hold uniformly in $\pi$) \\
Attack $A : \Theta \to (\text{guess})$; score $\hat b : \Theta \to [0,1]$
  & Decision rule $\varphi : [C]^T \to [0,1]$ \\
Membership test $H_0\!: \theta \sim M(\bar S)$ vs.\ $H_1\!: \theta \sim M(\bar S \cup \{z\})$
  & $H_0\!: \rva_T \sim \gM(\rmR_T')$ (silent) vs.\ $H_1\!: \rva_T \sim \gM(\rmR_T)$ (filed) \\
$\mathrm{adv}_{\mathrm{SMIA}} = \sup_\varphi (\mathrm{TPR} - \mathrm{FPR}) \le \eta$
  & $\mathrm{adv}_{\mathrm{SMIA}} \le \delta$ (Eq.~\ref{eq:prior-free}) \\
Bernoulli-prior success bound $\mathrm{succ} \le 1 - R_f(\pi)$ (\citep{kulynychunifying}, Thm.~D.1)
  & $\mathrm{adv}_\pi(\gM) \le \min\{\pi, 1-\pi\}\,\delta$;\ \ $\mathrm{adv}(\gM) \le \delta/2$ (Thm.~\ref{thm:advantage}) \\
\bottomrule
\end{tabular}
\end{table}

\section{Additional Proofs}
\label{app:sec:additional_proofs}

\begin{proposition}[TV characterization]
\label{app:prop:tv-char}
A mechanism $\gM$ satisfies $(0,\delta)$-DP if and
only if for every $\rmR_T \sim \rmR_T'$,
\[
  d_{\mathrm{TV}}\!\bigl(\gM(\rmR_T),\, \gM(\rmR_T')\bigr) \;\le\; \delta.
\]
\end{proposition}

\begin{proof}
Definition~\ref{def:DP} requires
$\Pr[\gM(\rmR_T) \in A] - \Pr[\gM(\rmR_T') \in A] \le \delta$ for every measurable $A$
and every adjacent pair. Taking the supremum over $A$, and using that $\sim$
is symmetric, yields exactly $d_{\mathrm{TV}}(\gM(\rmR_T), \gM(\rmR_T')) \le \delta$.
The converse is immediate from the definition of TV distance.
\end{proof}

\begin{lemma}[Group privacy for $k$ reports]
\label{app:lem:group_privacy_reports}
Suppose $\gM$ satisfies $(\epsilon,\delta)$-differential privacy with respect to
the single-report adjacency relation in Definition~\ref{def:adjacency}. Let
$\rmR_T$ and $\rmR'_T$ be report streams that differ in the inclusion or exclusion
of at most $k$ reports. Then, for all measurable $A \subseteq \gA_T$,
\[
  \Pr[\gM(\rmR_T) \in A]
  \leq
  e^{k\epsilon} \Pr[\gM(\rmR'_T) \in A]
  +
  \delta \sum_{i=0}^{k-1} e^{i\epsilon}.
\]
In particular, if $\epsilon = 0$, then $\gM$ satisfies $(0,k\delta)$-DP for
streams differing in at most $k$ reports.
\end{lemma}

\begin{proof}
Because $\rmR_T$ and $\rmR'_T$ differ in at most $k$ reports, there exists a
sequence of report streams
\[
  \rmR^{(0)}_T, \rmR^{(1)}_T, \ldots, \rmR^{(k)}_T
\]
such that $\rmR^{(0)}_T = \rmR_T$, $\rmR^{(k)}_T = \rmR'_T$, and each consecutive
pair $\rmR^{(j)}_T,\rmR^{(j+1)}_T$ is adjacent in the sense of
Definition~\ref{def:adjacency}. If the two original streams differ in fewer than
$k$ reports, we may repeat the final stream to obtain a sequence of length $k$.

Applying $(\epsilon,\delta)$-DP to the first adjacent pair gives
\[
  \Pr[\gM(\rmR^{(0)}_T) \in A]
  \leq
  e^\epsilon \Pr[\gM(\rmR^{(1)}_T) \in A] + \delta.
\]
Applying the same bound repeatedly yields
\[
\begin{aligned}
  \Pr[\gM(\rmR^{(0)}_T) \in A]
  &\leq e^{2\epsilon}\Pr[\gM(\rmR^{(2)}_T) \in A]
       + \delta(1+e^\epsilon) \\
  &\leq \cdots \\
  &\leq e^{k\epsilon}\Pr[\gM(\rmR^{(k)}_T) \in A]
       + \delta\sum_{i=0}^{k-1} e^{i\epsilon}.
\end{aligned}
\]
Substituting $\rmR^{(0)}_T=\rmR_T$ and $\rmR^{(k)}_T=\rmR'_T$ proves the claim.
For $\epsilon=0$, the sum is $k$, giving $(0,k\delta)$-DP.
\end{proof}

\begin{corollary}[Per-whistleblower advantage]
\label{app:cor:per_whistleblower}
Let $\rmR_T$ and $\rmR'_T$ differ in the inclusion or exclusion of all $k$
reports attributable to a single whistleblower, with no restriction on when
those reports were filed or which audit targets they concern. If
$\mathcal M$ satisfies $(0,\delta)$-differential privacy under single-report
adjacency (Definition~\ref{def:adjacency}), then the baseline-independent
advantage of any adversary in distinguishing $\rmR_T$ from $\rmR'_T$ using the
audit transcript is at most $k\delta/2$.
\end{corollary}
\begin{proof}
By group privacy (Lemma~\ref{app:lem:group_privacy_reports}),
$\mathcal M$ satisfies $(0,k\delta)$-DP for streams differing in at most $k$
reports. Applying Theorem~\ref{thm:advantage} with privacy parameter $k\delta$
therefore gives baseline-independent advantage at most $k\delta/2$.
\end{proof}

\begin{remark}[Relation of the advantage to the risk measures of \citet{kulynychunifying}]
\label{rem:kulynych-relation}
The weighted game defined in section~\ref{sec:threat_model} is the Bernoulli-prior instance of the strong
threat model of \citet{kulynychunifying}: their prior over the unknown
record specializes to the weight $\pi$ over
$\{\text{filed}, \text{silent}\}$, and Theorem~\ref{thm:advantage} is
their Theorem~D.1 evaluated at the trade-off function
$f(\alpha) = \max\{0,\, 1-\delta-\alpha\}$ of $(0,\delta)$-DP (stated there for replace-one; the add/remove case is identical). Their
prior-free SMIA advantage,
$\mathrm{adv}_{\mathrm{SMIA}} \coloneqq \sup\,(\mathrm{TPR} -
\mathrm{FPR})$, is bounded by $\delta$ via~\eqref{eq:prior-free}; the two summaries are related by
$\mathrm{adv}(\gM) = \tfrac{1}{2}\,\mathrm{adv}_{\mathrm{SMIA}}$,
with the supremum over $\pi$ attained at $\pi = 1/2$.
\end{remark}

\begin{remark}[Intermediate observation times]
\label{app:rem:prefix}
The transcript-level guarantee implies the same guarantee at every
intermediate time. Fix $t \le T$. Any event
$B \subseteq [C]^t$ concerning the prefix $\rva_{1:t} = (a_1, \dots, a_t)$
is the cylinder event $B \times [C]^{T-t}$ concerning the full
transcript, so for every adjacent pair $\rmR_T \sim \rmR'_T$,
\begin{align*}
\Pr\bigl[\gM(\rmR_T)_{1:t} \in B\bigr]
&= \Pr\bigl[\gM(\rmR_T) \in B \times [C]^{T-t}\bigr] \\
&\le \Pr\bigl[\gM(\rmR'_T) \in B \times [C]^{T-t}\bigr] + \delta \\
&= \Pr\bigl[\gM(\rmR'_T)_{1:t} \in B\bigr] + \delta,
\end{align*}
i.e., the prefix is itself $(0,\delta)$-DP. Moreover, every decision
rule $\psi : [C]^t \to [0,1]$ available to an adversary observing only
$\rva_{1:t}$ induces the transcript rule $\psi(\rva_{t})$ with
identical error rates, so the supremum in
Theorem~\ref{thm:advantage} over this restricted class is dominated by
the supremum over all rules: the advantage bound
$\min\{\pi, 1-\pi\}\,\delta$ holds verbatim at every $t \le T$. We
emphasize that these are not separate guarantees that compose: an
adversary observing the audit decisions as they unfold possesses, at
every moment, a deterministic function of the single release
$\rva_T$, so all intermediate-time bounds are instances of one
transcript-level budget $\delta$, and no accumulation over observation
times occurs. Since our mechanism releases decisions online
(Algorithms~\ref{alg:mechanism}), the prefix
$\rva_{t}$ is exactly the information available to the adversary at
time $t$; the guarantee on the full transcript is therefore the
strongest statement of this form.
\end{remark}

\section{Baselines}

\subsection{Greedy Auditing}
\label{app:subsec:greedy}

\begin{algorithm}[H]
\caption{Greedy Auditing Mechanism, $\gMg$}
\label{app:alg:greedy}
\begin{algorithmic}[1]
\Require Number of audit targets $C$
\State \textbf{Initialize.} For each $c \in [C]$, set $n_c \gets 0$.
\For{$t = 1, 2, 3, \ldots$}
    \State \textbf{Receive} new report counts $\rvr_t \in \mathbb{N}_0^C$.
    \For{each $c \in [C]$}
        \State $n_c \gets n_c + r_{t,c}$ \quad (active reports since last audit of $c$)
    \EndFor
    \State \textbf{Select} $a_t \gets \arg\max_{c \in [C]} n_c$
        \quad (ties broken uniformly at random).
    \State \textbf{Reset} $n_{a_t} \gets 0$.
\EndFor
\end{algorithmic}
\end{algorithm}




\begin{proposition}[No Nontrivial Privacy for Greedy Auditing]
\label{app:prop:greedy_no_privacy}
For $C \ge 2$, the greedy mechanism $\gMg$ in
Algorithm~\ref{app:alg:greedy} does not satisfy $(0,\delta)$-differential
privacy for any $\delta < 1 - 1/C$. 
Moreover, it does not satisfy
pure $\epsilon$-differential privacy for any finite $\epsilon$.
\end{proposition}

\begin{proof}
Consider a horizon containing a first audit decision, and let $C \ge 2$.
Let $\rmR'_T$ have no reports at time $1$, so that all audit targets are tied.
Let $\rmR_T$ be identical to $\rmR'_T$ except for one additional report for
target $2$ at time $1$. Then $\rmR_T \sim \rmR'_T$ by
Definition~\ref{def:adjacency}.

Under $\rmR_T$, target $2$ is the unique maximizer, so
\[
\Pr[\gMg(\rmR_T)_1 = 2] = 1.
\]
Under $\rmR'_T$, all $C$ targets are tied and ties are broken uniformly, so
\[
\Pr[\gMg(\rmR'_T)_1 = 2] = \frac{1}{C}.
\]
For the event $S := \{\rva : a_1 = 2\}$, $(0,\delta)$-differential privacy
therefore requires
\[
1 \leq \frac{1}{C} + \delta,
\]
and hence $\delta \geq 1 - 1/C$. Thus no nontrivial
$(0,\delta)$ guarantee with $\delta < 1 - 1/C$ is possible.

Moreover, pure $\epsilon$-differential privacy also fails for every finite
$\epsilon$. Let $S' := \{\rva : a_1 = 1\}$. Under $\rmR'_T$,
$\Pr[\gMg(\rmR'_T) \in S'] = 1/C$, whereas under $\rmR_T$,
$\Pr[\gMg(\rmR_T) \in S'] = 0$. Since the adjacency relation is symmetric,
pure $\epsilon$-DP would require
\[
\frac{1}{C} \leq e^\epsilon \cdot 0,
\]
which is impossible for finite $\epsilon$.
\end{proof}

\begin{proposition}[Perfect Utility of Greedy Auditing]
\label{app:prop:greedy_utility}
Let $n_{t, c}$ denote the active count of audit target $c$ at time $t$ as in Definition~\ref{def:active-counts}. The greedy mechanism in Algorithm~\ref{app:alg:greedy} satisfies:
\begin{enumerate}
    \item $a_t \in \arg\max_{c \in [C]} n_{t, c}$ with probability $1$.
    \item $\Pr[n_{t, a_t} = 0 \mid \exists c : n_{t, c} > 0] = 0$.
    \item If $c^\ast$ is the unique maximizer of active counts at time $t^\ast$,
    then $a_{t^\ast}=c^\ast$ with probability $1$.
\end{enumerate}
\end{proposition}

The proof follows by construction.

\subsection{Uniformly Random Auditing}
\label{app:sec:uniform}

\begin{algorithm}[H]
\caption{Uniformly Random Auditing, $\gMuni$}
\label{app:alg:uniformly_random}
\begin{algorithmic}[1]
\Require Number of audit targets $C$
\For{$t = 1, 2, 3, \ldots$}
    \State \textbf{Receive} new report counts $\rvr_t \in \N_0^C$ \quad (ignored)
    \State  \textbf{Select} $a_t \sim \text{Uniform}([C])$
\EndFor
\end{algorithmic}
\end{algorithm}

\begin{proposition}[Perfect Privacy of Uniformly Random Auditing]
\label{app:prop:uniform_privacy}
The uniformly random mechanism $\gMuni$ in Algorithm~\ref{app:alg:uniformly_random} satisfies $(0,0)$-differential privacy: for all
streams $\rmR_T,\rmR_T'$ of any length (adjacent or not) and all measurable $S$,
\[
\Pr[\gMuni(\rmR_T) \in S] = \Pr[\gMuni(\rmR_T') \in S].
\]
\end{proposition}


\begin{proof}
The output distribution is independent of the report stream. At every timestep,
the mechanism samples uniformly from $[C]$, regardless of the input. Therefore,
for any two streams $\rmR_T,\rmR_T'$, the induced distributions over audit transcripts
are identical.
\end{proof}


\begin{proposition}[Utility of Uniformly Random Auditing]
\label{app:prop:random_utility}
Let $n_{t, c}$ denote the active count of audit target $c$ at time $t$ as in Definition~\ref{def:active-counts}. At any time $t$:
\begin{enumerate}
    \item $\Pr[a_t = c] = 1/C$ for all $c \in [C]$, regardless of the active counts.
    \item Conditional on the decision prefix $\rva_{t-1}$, which fixes the
    active counts and hence the maximizer set
    $C_t^* = \arg\max_{c \in [C]} n_{t, c}$ (Definition~\ref{def:error}),
    \[
      \Pr\bigl[a_t \in C_t^* \,\big\vert\, \rva_{1:t-1}\bigr]
      = \frac{|C_t^*|}{C}.
    \]
\end{enumerate}
\end{proposition}

\begin{proof}
Items 1 and 2: conditional on the decision prefix $\rva_{1:t-1}$, the active
counts and hence $C_t^*$ are fixed, while $a_t \sim \mathrm{Uniform}([C])$
is drawn independently of the past; both identities follow by counting.
Item 3: the events $\{a_t = c^\ast\}$ for $t \ge t^\ast$ are independent,
each of probability $1/C$, and independent of the report stream. Hence
$t_{\mathrm{exit}} - t^\ast + 1$, the number of steps until $c^\ast$ is
first selected, is geometrically distributed with success probability
$1/C$, and its expectation is $C$.
\end{proof}

\subsection{Randomized Response}
\label{app:subsec:randomized_response}


\begin{algorithm}[t]
\caption{Randomized Response Auditing, $\gMrr$}
\label{app:alg:randomized_response}
\begin{algorithmic}[1]
\Require Number of audit targets $C$; horizon $T$; random-selection probability 
         $p_{\mathrm{rand}} \in (0,1]$
\State \textbf{Initialize.} For each $c \in [C]$, set the pending counts $n_c \gets 0$.
\For{$t = 1, 2, \ldots, T$}
    \For{each $c \in [C]$}
        \State \textbf{Receive} new report counts $r_{t,c} \in \N_0$
        \State $n_c \gets n_c + r_{t,c}$
    \EndFor
    \State Sample $X \sim \mathrm{Bernoulli}(p_{\mathrm{rand}})$
    \If{$X = 1$}
        \State $a_t \sim \mathrm{Uniform}(\{1,\ldots,C\})$ \Comment{explore}
    \Else
        \State $a_t \gets \argmax_{c \in [C]} n_c$ (ties broken uniformly at random) \Comment{exploit}
    \EndIf
    \State \textbf{Select} $a_t$
    \State \textbf{Reset} $n_{a_t} \gets 0$
\EndFor
\end{algorithmic}
\end{algorithm}

We first show the following sufficient condition for $(0, \delta)$-DP:

\begin{proposition}[Privacy of Randomized Response]
\label{app:prop:RR_privacy}
Fix a horizon $T \in \mathbb{N}$. Algorithm~\ref{app:alg:randomized_response}
satisfies $(0, \delta)$-differential privacy at horizon $T$ for
$\delta := 1 - p_{\mathrm{rand}}^T$. This is an \emph{upper bound} on the privacy loss.
\end{proposition}

\begin{proof}
We show that, for all adjacent report streams
$\mathbf{R}_T \sim \mathbf{R}'_T$ and all measurable
$A \subseteq \gA_T$,
\begin{equation}
\begin{split}
  \Pr[\gMrr(\mathbf{R}_T) \in A]
 \le
  \Pr[\gMrr(\mathbf{R}'_T) \in A] + 1 - p_{\mathrm{rand}}^T .
\end{split}
\label{eq:rr-privacy-goal}
\end{equation}

Let $U_1,\dots,U_T \overset{\mathrm{iid}}{\sim}
\mathrm{Bernoulli}(p_{\mathrm{rand}})$ be the random coins drawn by the
algorithm, independently of the report stream. Define the
\emph{full-exploration event}
\[
  E := \{U_t = 1 \text{ for all } t \in [T]\},
  \qquad
  \Pr[E] = p_{\mathrm{rand}}^T .
\]

\paragraph{Step 1: The output is data-free on $E$.}
On $E$, every round satisfies $U_t=1$, so
Algorithm~\ref{app:alg:randomized_response} always explores and outputs
$a_t \sim \mathrm{Uniform}(\{1,\dots,C\})$ independently of the report
counts. Thus the conditional output distribution is the same under
$\mathbf{R}_T$ and $\mathbf{R}'_T$. Hence, for any $A$,
\begin{equation}
\begin{split}
  \Pr[\gMrr(\mathbf{R}_T)\in A \mid E]
  &=
  \Pr[\gMrr(\mathbf{R}'_T)\in A \mid E] \\
  &:= P_0(A).
\end{split}
\label{eq:same-on-E}
\end{equation}

\paragraph{Step 2: Total-probability decomposition.}
For compactness, define
\[
  Q_T(A)
  :=
  \Pr[\gMrr(\mathbf{R}_T)\in A \mid \neg E],
\]
and
\[
  Q_T'(A)
  :=
  \Pr[\gMrr(\mathbf{R}'_T)\in A \mid \neg E].
\]
Conditioning on whether $E$ holds, and using eq.~(\ref{eq:same-on-E}), gives
\begin{align}
  \Pr[\gMrr(\mathbf{R}_T)\in A]
  &=
  p_{\mathrm{rand}}^T P_0(A)
  +
  (1-p_{\mathrm{rand}}^T) Q_T(A),
  \label{eq:rr-decomp-R}
  \\
  \Pr[\gMrr(\mathbf{R}'_T)\in A]
  &=
  p_{\mathrm{rand}}^T P_0(A)
  +
  (1-p_{\mathrm{rand}}^T) Q_T'(A).
  \label{eq:rr-decomp-Rprime}
\end{align}
Subtracting eq.~(\ref{eq:rr-decomp-Rprime}) from eq.~(\ref{eq:rr-decomp-R})
yields the exact equality
\begin{equation}
\begin{split}
&\Pr[\gMrr(\mathbf{R}_T)\in A]
 -
 \Pr[\gMrr(\mathbf{R}'_T)\in A]  =
(1-p_{\mathrm{rand}}^T)\Delta_T(A),
\end{split}
\label{eq:exact-decomp}
\end{equation}
where $\Delta_T(A)\coloneqq
  Q_T(A) - Q_T'(A)$.

\paragraph{Step 3: Bounding the conditional gap.}
Since $Q_T(A)$ and $Q_T'(A)$ are probabilities,
\[
  \Delta_T(A)
  =
  Q_T(A) - Q_T'(A)
  \le 1.
\]

\paragraph{Conclusion.}
Combining eq.~(\ref{eq:exact-decomp}) with Step~3 gives
\[
  \Pr[\gMrr(\mathbf{R}_T)\in A]
  \le
  \Pr[\gMrr(\mathbf{R}'_T)\in A]
  + 1 - p_{\mathrm{rand}}^T .
\]
This holds for every $\mathbf{R}_T\sim\mathbf{R}'_T$ and every
$A$, so Algorithm~\ref{app:alg:randomized_response} satisfies
$(0,\,1-p_{\mathrm{rand}}^T)$-DP at horizon $T$.
\end{proof}

We further want to prove the following statement on the error bound.



\begin{proposition}[One-step error of randomized response]
\label{app:prop:RR_error}
Fix a horizon $T \in \mathbb{N}$ and privacy parameter $\delta \in [0,1)$.
A sufficient privacy calibration from Proposition~\ref{app:prop:RR_privacy} is
\[
  p_{\mathrm{rand}} = (1-\delta)^{1/T}.
\]
Fix any time $t \le T$ and condition on any pre-decision history
$\mathcal H_{t-1}=h$. Let $n_{t, c}$ denote the active count of audit target
$c$ at time $t$ (Definition~\ref{def:active-counts}). If $c_t^*$ is the unique
maximizer of the active counts, i.e.
\[
  c_t^* = \arg\max_{c\in[C]} n_{t, c},
\]
then Algorithm~\ref{app:alg:randomized_response} satisfies
\[
  \Pr[a_t \neq c_t^* \mid \mathcal H_{t-1}=h]
  =
  \frac{C-1}{C}\,p_{\mathrm{rand}}.
\]
In particular, under the above sufficient privacy calibration,
\[
  \Pr[a_t \neq c_t^* \mid \mathcal H_{t-1}=h]
  =
  \frac{C-1}{C}(1-\delta)^{1/T}.
\]
\end{proposition}

\begin{remark}[Why this parameterisation?]
By Proposition~\ref{app:prop:RR_privacy}, Algorithm~\ref{app:alg:randomized_response}
satisfies $(0,\delta)$-DP at horizon $T$ whenever
\[
  1 - p_{\mathrm{rand}}^T \le \delta,
  \qquad\text{equivalently}\qquad
  p_{\mathrm{rand}} \ge (1-\delta)^{1/T}.
\]
Thus the smallest exploration probability certified by
Proposition~\ref{app:prop:RR_privacy} is
\[
  p_{\mathrm{rand}} = (1-\delta)^{1/T}.
\]
Since the one-step error in Proposition~\ref{app:prop:RR_error} increases linearly
with $p_{\mathrm{rand}}$, this is the error-minimising choice within this
sufficient privacy calibration. For fixed $\delta$, however,
$(1-\delta)^{1/T}\to 1$ as $T\to\infty$, so the mechanism approaches uniformly
random auditing over long horizons.
\end{remark}

\begin{proof}[Proof of Proposition~\ref{app:prop:RR_error}]
Fix a time $t \le T$ and condition on a pre-decision history
$\mathcal H_{t-1}=h$ such that the active counts
$\{n_{t, c}\}_{c\in[C]}$ are fixed and have a unique maximizer $c_t^*$.
Let $U_t \sim \mathrm{Bernoulli}(p_{\mathrm{rand}})$ be the random coin drawn
by the algorithm at round $t$, independently of the report stream and of the
past randomness.

By the law of total probability,
\begin{equation}
\begin{split}
  \Pr[a_t \neq c_t^* \mid \mathcal H_{t-1}=h]
  &=
  \Pr[U_t=1]\Pr[a_t \neq c_t^* \mid U_t=1,\mathcal H_{t-1}=h] \\
  &\quad+
  \Pr[U_t=0]\Pr[a_t \neq c_t^* \mid U_t=0,\mathcal H_{t-1}=h].
\end{split}
\label{app:eq:rr-error-total-prob}
\end{equation}

If $U_t=1$, the algorithm explores and draws
$a_t \sim \mathrm{Uniform}([C])$. Since exactly one audit target equals
$c_t^*$,
\[
  \Pr[a_t \neq c_t^* \mid U_t=1,\mathcal H_{t-1}=h]
  =
  \frac{C-1}{C}.
\]
If $U_t=0$, the algorithm exploits and chooses
\[
  a_t \in \arg\max_{c\in[C]} n_{t, c}.
\]
By assumption, $c_t^*$ is the unique maximizer, so
\[
  \Pr[a_t \neq c_t^* \mid U_t=0,\mathcal H_{t-1}=h] = 0.
\]
Substituting into eq.~(\ref{app:eq:rr-error-total-prob}) gives
\[
  \Pr[a_t \neq c_t^* \mid \mathcal H_{t-1}=h]
  =
  p_{\mathrm{rand}}\frac{C-1}{C}.
\]
Finally, using the sufficient privacy calibration
$p_{\mathrm{rand}}=(1-\delta)^{1/T}$ gives
\[
  \Pr[a_t \neq c_t^* \mid \mathcal H_{t-1}=h]
  =
  \frac{C-1}{C}(1-\delta)^{1/T}.
\]
\end{proof}

\subsubsection{Proof of Theorem~\ref{thm:RR_tradeoff} (Privacy of Randomized Response}
\label{app:subsubsec:proof_rr}

We now prove the main result of section~\ref{sec:rand_response}:


\begin{proof}[Proof]

\medskip\noindent\textbf{Step 1: Reduction to the first decision.} 
By Remark~\ref{rem:prefix} (the cylinder inequality with $t = 1$),
$(0,\delta)$-DP of the transcript implies, for every adjacent pair
$\rmR_T \sim \rmR'_T$ and every $c \in [C]$,
\[
    \Pr[a_1 = c \mid \rmR_T] \;\le\; \Pr[a_1 = c \mid \rmR'_T] + \delta.
\]

\medskip\noindent\textbf{Step 2: An adjacent pair.} 
Let $\rmR'_T$ be the all-zero
stream, and let $\rmR_T$ be identical except for a single report
$r_{1,c^*} = 1$ for an arbitrary $c^* \in [C]$. The pair is adjacent
under Definition~\ref{def:adjacency}.
 
\medskip\noindent\textbf{Step 3: First-decision laws.} 
Under $\rmR_T$, the pending counts at $t=1$ are $1$ at $c^*$ and $0$ elsewhere, so $c^*$ is the unique maximizer: the exploit branch (probability $1 - p_{\mathrm{rand}}$) selects $c^*$ with certainty, and the explore branch selects it with probability $1/C$, giving
\[
    \Pr[a_1 = c^* \mid \rmR_T]
    = (1 - p_{\mathrm{rand}}) + \frac{p_{\mathrm{rand}}}{C}.
\]
Under $\rmR'_T$, all counts are zero, the exploit branch breaks the
$C$-way tie uniformly, and the explore branch is uniform, giving
$\Pr[a_1 = c^* \mid \rmR'_T] = 1/C$.
 
\medskip\noindent\textbf{Step 4: Apply the privacy constraint.} 
Applying Step 1 to the event $\{a_1 = c^*\}$,
\[
    (1 - p_{\mathrm{rand}})\Bigl(1 - \frac{1}{C}\Bigr)
    = \Pr[a_1 = c^* \mid \rmR_T] - \Pr[a_1 = c^* \mid \rmR'_T]
    \;\le\; \delta,
\]
which rearranges to
$p_{\mathrm{rand}} \ge 1 - \delta C/(C-1)$.
 
\medskip\noindent\textbf{Step 5: Error consequence.} 
By Proposition~\ref{app:prop:RR_error}, the conditional error at any
unique-maximizer history equals $\frac{C-1}{C} p_{\mathrm{rand}}$; substituting the bound of Step~4 gives $\frac{C-1}{C} p_{\mathrm{rand}} \ge \frac{C-1}{C}\bigl(1 -\frac{\delta C}{C-1}\bigr) = \frac{C-1}{C} - \delta$.
\end{proof}
 
 
\begin{proposition}[Horizon-independent privacy of randomized response]
\label{app:prop:RR_horizon_free}
Fix $C \ge 2$ and $p_{\mathrm{rand}} \in (0,1]$. For every horizon
$T \in \mathbb{N}$, Algorithm~\ref{app:alg:randomized_response}
satisfies $(0, \delta_p)$-DP with
\[
    \delta_p \;=\; \min\Bigl\{\, 1 - p_{\mathrm{rand}}^{\,T},\;
    \frac{C\,(1-p_{\mathrm{rand}})}{p_{\mathrm{rand}}} \,\Bigr\}.
\]
In particular, $p_{\mathrm{rand}} \ge C/(C+\delta)$ implies
$(0,\delta)$-DP at every horizon simultaneously.
\end{proposition}
 
\begin{proof}
The first term is Proposition~\ref{app:prop:RR_privacy}; we prove the
second by a coupling argument. Write $p = p_{\mathrm{rand}}$. Fix an
adjacent pair differing at $(t^*, c^*)$; by symmetry of total
variation we may assume $\rmR_T$ contains the extra report. Couple the
two executions on shared randomness: the same explore/exploit
indicators $U_1, \dots, U_T \sim \mathrm{Bern}(p)$, the same
exploration targets $V_1, \dots, V_T \sim \mathrm{Unif}[C]$, and, at
each exploit step, a maximal coupling of the two worlds' selection
distributions given their current pending counts. Let $(a_t)$ and
$(a'_t)$ denote the coupled transcripts.
 
\medskip\noindent\textbf{Claim 1: State tracking.} On the event that $a_s = a'_s$ for
all $s < t$, the two worlds' pending counts at the decision at time
$t$ agree in every coordinate $c \ne c^*$, and differ at $c^*$ by
exactly $1$ if $t \ge t^*$ and $c^*$ has not been audited at any
$s \in [t^*, t)$, and by $0$ otherwise. This follows by induction on
$t$: identical decisions imply identical resets, all report inputs
agree except the single extra report at $(t^*, c^*)$, and an audit of
$c^*$ at any $s \ge t^*$ resets both worlds' $c^*$-counts to zero
simultaneously, after which all inputs agree.
 
\medskip\noindent\textbf{Claim 2: No divergence before $t^*$.} For $t < t^*$ the counts
agree in all coordinates (Claim 1), so at exploit steps the two
selection distributions are identical and the maximal coupling makes
the decisions agree with probability $1$; at explore steps both select
$V_t$. Hence the first disagreement, if any, occurs at some
$t \ge t^*$.
 
\medskip\noindent\textbf{Claim 3: Explore-hits force permanent agreement.} Define the
explore-hit event $E_s := \{U_s = \text{explore},\, V_s = c^*\}$, with
$\Pr[E_s] = p/C$, independent across $s$ and of all other randomness.
If no disagreement has occurred before $s \ge t^*$ and $E_s$ occurs,
then both worlds audit $c^*$ at step $s$; by Claim 1 their counts
coincide in all coordinates from step $s{+}1$ onward, all subsequent
inputs agree, and by the argument of Claim 2 the decisions agree at
every later step. The transcripts then never differ.
 
\medskip\noindent\textbf{Claim 4: Divergence requires exploitation.} At an explore step
both worlds select $V_t$, so a first disagreement at step $t$ requires
$U_t = \text{exploit}$.
 
Combining Claims 2--4: if the first disagreement occurs at step $t$,
then $t \ge t^*$, $U_t = \text{exploit}$, and no explore-hit $E_s$
occurred for $s \in [t^*, t)$ (else, by Claim 3, disagreement would be
impossible at $t$). Since these are events over independent coins,
\[
    \Pr[\text{first disagreement at } t]
    \;\le\; (1-p)\,\Bigl(1 - \frac{p}{C}\Bigr)^{t - t^*},
\]
and summing over $t \ge t^*$,
\[
    \Pr[\text{transcripts ever differ}]
    \;\le\; (1-p) \sum_{k=0}^{\infty} \Bigl(1 - \frac{p}{C}\Bigr)^{k}
    \;=\; \frac{C\,(1-p)}{p}.
\]
By the coupling characterization of total variation,
$d_{\mathrm{TV}}(\gM(\rmR_T), \gM(\rmR'_T)) \le
C(1-p)/p$ for every adjacent pair, which by
Proposition~\ref{app:prop:tv-char} gives the claimed
$(0, C(1-p)/p)$-DP at every horizon. The final statement follows since
$C(1-p)/p \le \delta \iff p \ge C/(C+\delta)$.
\end{proof}
 
\begin{remark}[The reset is essential]
The proof of Proposition~\ref{app:prop:RR_horizon_free} uses
Condition~\ref{cond:reset} in Claim~1: auditing $c^*$ removes the
differing report from both worlds' states. For a variant of randomized
response running on cumulative counts without resets, the influence of
a single report persists past audits of its target, and a
horizon-independent guarantee of this form need not hold.
\end{remark}

\section{Proof of Proposition~\ref{prop:privacy_reduction} (Privacy of Algorithm~\ref{alg:mechanism})}
\label{app:privacy_reduction_proof}

We give the full proof of Proposition~\ref{prop:privacy_reduction}; the sketch appears in Section~\ref{sec:privacy}.

\begin{proof}
 Let $\rmR_T \sim \rmR_T'$ be any adjacent pair, differing at $(t^*, c^*)$ with
  $r'_{t^*,c^*} = r_{t^*,c^*} + 1$ and $r_{t,c} = r'_{t,c}$ for all
  $(t,c) \ne (t^*,c^*)$.  The proof has four steps.

  \medskip\noindent
  \textbf{Step~1: The transcript is a deterministic function of counter outputs.}

  At each step $t$, the decision $a_t = \arg\max_c \tilde{n}_{t,c}$ is a
  deterministic function of $\{\tilde{n}_{t,c}\}_{c \in [C]}$, which are outputs
  of the counter instances.  By induction over $t$, the full transcript
  $(a_1, \ldots, a_T)$ is a deterministic function $g$ of the collection of all
  outputs produced by all counter instances across all steps.

  \medskip\noindent
  \textbf{Step~2: Identify the unique counter instance affected by the extra
  report.}

  We show by induction on $t$ that for all $t < t^*$, all counter outputs at
  step $t$ are \emph{the same random variables} under $\rmR_T$ and $\rmR_T'$.

  \emph{Base case} ($t = 1$, or vacuously if $t^* = 1$): at $t=1$ all instances
  are freshly initialized with randomness drawn independently of all data.  Since
  $r_{1,c} = r'_{1,c}$ for every $c$ (the streams agree at $t=1$ if $t^* > 1$),
  all counter outputs at $t=1$ are identical under both streams.

  \emph{Inductive step}: assume all counter outputs at steps $1, \ldots, t-1$
  are identical under $\rmR_T$ and $\rmR_T'$.  Then $a_{t-1}$ is the same under both
  streams.  Any restart at step $t-1$ produces a fresh instance with randomness
  drawn independently of all data; since the restart event is identical, this
  instance is the same random object under both streams.  At step $t$, each
  $\gC_c$ receives the same input $r_{t,c} = r'_{t,c}$ (as $t < t^*$) and has
  the same internal state, so its output is identical.  The induction goes
  through.

  It follows that the audit history $\calH_{t^*-1} = (a_1, \ldots, a_{t^*-1})$
  is \emph{identical} under $\rmR_T$ and $\rmR_T'$.  In particular, the restart history
  of $c^*$ up to $t^*-1$ is the same, so the counter instance $\gC_{c^*}^*$
  that is active for $c^*$ at time $t^*$ --- started at some step
  $s^* \le t^*$ --- is the same object (same randomness, started at the same
  step) under both streams.

  For $c \ne c^*$: reports are identical under $\rmR_T$ and $\rmR_T'$ at all times, and
  restart decisions (driven by $\calH$) are the same.  Hence every instance of
  $\gC_c$ is the same random variable under both streams.

  For $c = c^*$: every instance of $\gC_{c^*}$ \emph{other than}
  $\gC_{c^*}^*$ either ran entirely before $t^*$ (identical inputs) or starts
  fresh after some audit of $c^*$ at time $s > t^*$ (with reports identical for
  all $t \ne t^*$, and identical restart times conditional on $\gC_{c^*}^*$
  agreeing --- see Step~4).  In either case its inputs do not include the extra
  report.

  The \emph{only} instance whose input stream differs between $\rmR_T$ and $\rmR_T'$ is
  $\gC_{c^*}^*$: it receives $r_{t,c^*}$ under $\rmR_T$ and
  $r_{t,c^*} + \mathbf{1}[t=t^*]$ under $\rmR_T'$ at each step of its run, so its
  input streams are adjacent (differing at position $j^* = t^* - s^* + 1 \le T$).

  \medskip\noindent
  \textbf{Step~3: Apply the DP guarantee of $\gC$.}

  Since $\gC_{c^*}^*$ processes a stream of length at most $T$ that is adjacent
  between $\rmR_T$ and $\rmR_T'$, and $\gC$ satisfies $(0,\delta)$-DP:
  \[
    d_{\mathrm{TV}}\!\bigl(
      \text{outputs of } \gC_{c^*}^* \text{ under } R,\;
      \text{outputs of } \gC_{c^*}^* \text{ under } R'
    \bigr)
    \;\le\; \delta.
  \]

  \medskip\noindent
  \textbf{Step~4: Couple all counter outputs and conclude.}

  Couple all counter instances identically except $\gC_{c^*}^*$ (possible since
  they are the same random variables under both streams).  For $\gC_{c^*}^*$,
  use the \emph{optimal coupling}, which achieves
  $\Pr[\gC_{c^*}^*\text{ disagrees}] \le \delta$.

  Conditional on $\gC_{c^*}^*$ agreeing: $a_{t^*}$ is the same under both
  streams (it is a deterministic function of the noisy counts, all of which are
  now identical).  Therefore all subsequent restarts occur at the same times, all
  subsequent counter inputs are identical (reports agree everywhere except at
  $t^*$, already handled), and all subsequent counter outputs can be coupled
  identically.  Hence:
  \[
    \Pr[\text{any counter output disagrees}] \;\le\; \delta.
  \]
  Since $g$ is deterministic, agreement of all counter outputs implies agreement
  of transcripts.  By the coupling characterisation of TV distance:
  \[
    d_{\mathrm{TV}}\!\bigl(\calM(R),\, \calM(R')\bigr)
    \;\le\;
    \Pr[\text{transcripts disagree}]
    \;\le\; \delta.
  \]
  Since $d_{\mathrm{TV}}(P,Q) = \sup_\calA |\Pr[P \in \calA] - \Pr[Q \in \calA]|$,
  this gives $(0,\delta)$-DP for all adjacent pairs. 
\end{proof}

\section{Proof of Proposition~\ref{prop:base-counter-privacy}}
\label{app:base_counter_privacy_proof}

Before proving Proposition~\ref{prop:base-counter-privacy} we state a useful textbook Lemma.

\begin{lemma}[TV distance for shifted Gaussians]
\label{app:lem:tv-shifted}
If
\[
P=\mathcal N(\boldsymbol{\mu},\sigma^2 \rmI_n),
\qquad
Q=\mathcal N(\boldsymbol{\mu}+\Delta,\sigma^2 \rmI_n),
\]
then
\[
d_{\mathrm{TV}}(P,Q)
=
2\Phi\!\left(\frac{\|\Delta\|_2}{2\sigma}\right)-1,
\]
where $\Phi$ is the standard Gaussian CDF.
\end{lemma}

\begin{proof}[Proof of Proposition~\ref{prop:base-counter-privacy}]
Let $\rvx\sim \rvx'$ be adjacent, differing only in coordinate $j^*$.  Then
$\rvx'-\rvx=\rve_{j^*}$, so the encoded noisy vectors satisfy
\[
\rmC^{(T)}\rvx' + \rvz
=
\rmC^{(T)}\rvx + \rvz + \rmC^{(T)}\rve_{j^*}.
\]
Thus the two encoded distributions are Gaussians with common covariance
$\sigma^2 \rmI_T$ and mean shift $\Delta = \rmC^{(T)}\rve_{j^*}$.  By
Lemma~\ref{app:lem:tv-shifted},
\begin{align*}
d_{\mathrm{TV}}\left(\rmC^{(T)}\rvx' + \rvz, \rmC^{(T)}\rvx + \rvz \right)
 =
2\Phi\!\left(\frac{\|\rmC^{(T)}\rve_{j^*}\|_2}{2\sigma}\right)-1 \le 
2\Phi\!\left(\frac{M_T}{2\sigma}\right)-1
=
\delta.
\end{align*}
Since $\gC(\rvx)$ is obtained from the encoded noisy vector by
deterministic post-processing through $\rmB^{(T)}$, total variation cannot increase.
Hence $\gC$ is $(0,\delta)$-\textit{DP}.
\end{proof}

 \section{Proof of Theorem~\ref{thm:utility} (Utility)}
\label{app:subsec:utility}


The proof of Theorem~\ref{thm:utility} in Section~\ref{sec:utility-thm} relies on the following Lemma.

\begin{lemma}[Fresh-noise decomposition]
\label{app:lem:fresh}
Fix any time $t \in [T]$ and any audit target $c \in [C]$. Let $\ell:=\lct \in [T]$ be the length of the active run of $c$ at time $t$ and define $\tau_c:=t-\lct+1$. 
For simplicity, we re-number and write
$\mathbf{x}^{(t,c)}:=(x^{(t,c)}_1,\dots,x^{(t,c)}_{\lct})$ for the report counts in that run, and
$\boldsymbol{\zeta}^{(t,c)}:=(\zeta^{(t,c)}_1,\dots,\zeta^{(t,c)}_{\lct})$ for the Gaussian noise vector associated with the
active instance of the Gaussian matrix counter instantiated by the $\lct \times \lct$ leading principal truncations $\rmB^{(\lct)}$ and $\rmC^{(\lct)}$. 

Define
\begin{equation}
\label{eq:Gct}
  \Gct
  :=
  \sum_{j=1}^{\lct-1} f_{\lct-j} \zeta^{(t,c)}_j,
\end{equation}
with  $\Gct:=0$ if $\lct=1$. Then the noisy pending count satisfies
\begin{equation}
\label{eq:decomp-ct}
  \tilde{n}_{t, c}
  =
  \underbrace{n_{t, c}+\Gct}_{=:\mu_{t, c}}
  +
  \zeta^{(t, c)}_{\lct},
\end{equation}
where $n_{t, c}=\sum_{j=1}^{\lct} x^{(t, c)}_j$ is the count of currently active reports. Moreover:
\begin{enumerate}
  \item[(i)] conditionally on $\mathcal{H}_{t-1} = h$, the fresh coordinate
$\zeta^{(t, c)}_{\ell_{t, c}}$ is distributed $\mathcal{N}(0,\sigma^2)$;
  \item[(ii)] for $c\neq c'$, the variables
  $\zeta^{(t, c)}_{\lct}$ and $\zeta^{(t,c')}_{\ell_{(t, c')}}$
  are independent
\end{enumerate}
\end{lemma}

\begin{proof}
Since the active run has length $\lct$, the current counter output is the
$\lct$-th output of the truncated mechanism, i.e.
\[
\tilde n_{t, c}
=
n_{t, c} + (\rmB^{(\lct)} \boldsymbol{\zeta}^{(t, c)})_{\lct}
\]
Because $\rmB^{(\lct)}$ is lower-triangular Toeplitz with
\[
\rmB^{(\lct)}_{\lct,j}=f_{\lct-j}
\qquad\text{for } 1\le j\le \lct,
\]
we obtain
\begin{align}
(\rmB^{(\lct)} \boldsymbol{\zeta}^{(t, c)})_{\lct}
&=
\sum_{j=1}^{\lct} f_{{\lct}-j} \zeta^{(t,c)}_j \notag\\
&=
f_0 \zeta^{(t, c)}_{\lct} + \sum_{j=1}^{\lct-1} f_{\lct-j} \zeta^{(t, c)}_{j} \notag\\
&=
\zeta^{(t, c)}_{\lct} + G_{t, c}
\label{eq:toeplitz-expand}
\end{align}
using $f_0=1$ and the definition of $G_{t, c}$. This proves
eq.~(\ref{eq:decomp-ct}).

For part (i), every earlier output of the same active instance corresponds
to some run-position $\ell <\lct$, and therefore depends only on
$\zeta^{(t, c)}_1,\dots,\zeta^{(t, c)}_{l}$, hence only on
$\zeta^{(t, c)}_1,\dots,\zeta^{(t, c)}_{\lct-1}$. Thus the coordinate $\zeta^{(t, c)}_{\lct}$
has not appeared in any previous output of the active instance. Since the
coordinates of $\zeta^{(t, c)}_{\lct}$ are i.i.d.\ Gaussian, $\zeta^{(t, c)}_{\lct}$ is
independent of $\zeta^{(t, c)}_1,\dots,\zeta^{(t, c)}_{\lct-1}$ and therefore
independent of $\mathcal H_{t-1}$.

For part (ii), if $c\neq c'$, then the active instances for $c$ and $c'$
use independently sampled Gaussian vectors. Hence the fresh coordinates
$\zeta^{(t, c)}_{\lct}$ and $\zeta^{(t, c')}_{\lct}$ are independent.
\end{proof}

\begin{proof}[Proof of Theorem~\ref{thm:utility}]

Throughout, condition on $\calH_{t-1} = h$, which fixes $\ell_c$, $n_{t, c}$,
$G_c$, $\mu_{t, c} = n_{t, c} + G_c$, and $c^*_t$.  The only remaining randomness is
$\{\zeta^{(t, c)}_{\lct}\}_{c \in [C]}$.

\medskip\noindent
\textbf{Step~1: Noisy counts in terms of fresh noise.}
By Lemma~\ref{app:lem:fresh}, $\tilde{n}_{t, c} = \mu_{t, c} + \zeta^{(t, c)}_{\lct}$
where $\{\zeta^{(t, c)}_{\lct}\}_{c \in [C]}$ are i.i.d.\ $\mathcal{N}(0,\sigma^2)$,
independent of $h$.  Hence
$a_t = \arg\max_{c}(\mu_{t, c} + \zeta^{(t, c)}_{\lct})$.

\medskip\noindent
\textbf{Step~2: Union bound over challengers.}
\[
  \bigl\{a_t \ne c^*_t\bigr\}
  \;\subseteq\;
  \bigcup_{c\,\ne\,c^*_t}
  \bigl\{
    \zeta^{(t, c)}_{\lct} -  \zeta^{(t, c^*_t)}_{\ell_{t, c^*_t}}
    \;\ge\;
    \mu_{c^*_t} - \mu_{t, c}
  \bigr\}.
\]
By the union bound:
\begin{equation}\label{eq:union}
  \Pr\bigl[a_t \ne c^*_t \mid h\bigr]
  \;\le\; \sum_{c\,\ne\,c^*_t}
  \Pr\!\bigl[W_c \ge \tilde{\Delta}_{t, c} \mid h\bigr],
\end{equation}
where $W_c := \zeta^{(t, c)}_{\lct} -  \zeta^{(t, c^*_t)}_{\ell_{t, c^*_t}}$.

\medskip\noindent
\textbf{Step~3: Exact probability for each challenger.}
By Lemma~\ref{app:lem:fresh}(i)--(ii), $\zeta^{(t, c)}_{\lct}$ and
$ \zeta^{(t, c^*_t)}_{\ell_{t, c^*_t}}$ are independent $\mathcal{N}(0,\sigma^2)$
variables, independent of $h$.  Their difference satisfies
$W_c \sim \mathcal{N}(0, 2\sigma^2)$, so
\begin{equation}\label{eq:exact}
  \Pr\bigl[W_c \ge \tilde{\Delta}_{t, c} \mid h\bigr]
  \;=\; \Phi\!\left(-\frac{\tilde{\Delta}_{t, c}}{\sqrt{2}\,\sigma}\right).
\end{equation}
This is an \emph{equality}, not an inequality.

\medskip\noindent
\textbf{Step~4: Substitute the noise calibration.}
With $\sigma = M_T/(2\kappa_\delta)$, we have $\sqrt{2}\,\sigma =
M_T/(\sqrt{2}\,\kappa_\delta)$, so
$\tilde{\Delta}_{t,c}/(\sqrt{2}\,\sigma) = \sqrt{2}\,\kappa_\delta\,\tilde{\Delta}_{t,c}/M_T$.
Substituting into~ eq.~(\ref{eq:union}) via eq.~(\ref{eq:exact}) gives the conditional
bound eq.~(\ref{eq:cond-utility}).  Taking expectations and applying the law of
total expectation gives the desired result. 
\end{proof}

\begin{remark}[A sufficient condition for small unconditional error]
\label{app:rem:sufficient}
A condition on true pending counts (before noise) can be given as follows.
Let
\[
  \mathsf{ME}_T
  :=
  \mathrm{MaxErr}(\rmB^{(T)},\rmC^{(T)})
  =
  M_T^2 .
\]
For any $c \ne c^*_t$, the accumulated-noise difference
$G_{t,c^*_t} - G_{t,c}$ is mean-zero Gaussian with variance
$V \le 2\sigma^2 M_T^2$. Therefore,
\begin{align}
  \Pr\!\left[
    \tilde{\Delta}_{t,c}
    \le \frac{\Delta_{t,c}}{2}
  \right]
  &=
  \Pr\!\left[
    G_{t,c^*_t} - G_{t,c}
    \le -\frac{\Delta_{t, c}}{2}
  \right] \notag \\
  &\le
  \exp\!\left(
    -\frac{\Delta_{t, c}^2}
           {16\sigma^2 M_T^2}
  \right) \notag \\
  &=
  \exp\!\left(
    -\frac{\kappa_\delta^2\Delta_{t,c}^2}
           {4\mathsf{ME}_T^2}
  \right).
\end{align}
Here the inequality uses the one-sided Gaussian tail bound
$\Pr[Y \le -\lambda] \le \exp(-\lambda^2/(2V))$ for
$Y \sim \mathcal{N}(0,V)$, and the final equality uses
$\sigma = M_T/(2\kappa_\delta)$.

By total probability, for each challenger $c$,
\begin{align}
  \Pr[\tilde{n}_{t,c} \ge \tilde{n}_{t,c^*_t}]
  &\le
  \Phi\!\left(
    -\frac{\kappa_\delta\Delta_{t,c}}
           {\sqrt{2}M_T}
  \right)
  +
  \exp\!\left(
    -\frac{\kappa_\delta^2\Delta_{t,c}^2}
           {4\mathsf{ME}_T^2}
  \right).
\end{align}
Both terms decrease rapidly as $\Delta_{t,c}$ grows, giving an
unconditional error bound in terms of the true pending-count advantage
alone.
\end{remark}

\begin{remark}[Interpretation of the effective gap $\tilde{\Delta}_{t,c}$]
\label{rem:eff-gap}
The effective gap decomposes as $\tilde{\Delta}_{t,c} = \Delta_{t,c} +
(G_{t,c^*_t} - G_{t,c})$, where $\Delta_{t,c} > 0$ is the true pending-count
advantage and $G_{t,c^*_t} - G_{t,c}$ is the difference of accumulated past
noise terms. The latter has mean zero and variance
$\sigma^2\bigl(\sum_{k=1}^{\ell_{c^*_t}-1}f_k^2 + \sum_{k=1}^{\ell_{t,c}-1}f_k^2\bigr)
\le 2\sigma^2 M_T^2$. When $\Delta_{t,c}$ is large relative to $\sigma M_T$,
the accumulated noise is unlikely to reverse the ordering and the effective
gap is typically close to $\Delta_{t,c}$.
\end{remark}
 
\begin{remark}[Tightness of the error bound]
\label{rem:tight}
The only inequality in the proof is the union bound (Step~2), which is tight
when one challenger dominates the sum, for example when $C=2$ (a single
challenger, in which case the bound is an equality) or when one challenger
has a much smaller effective gap than all others. The per-challenger
probabilities eq.~(\ref{eq:exact}) are equalities. The bound is therefore the
tightest of this form.
\end{remark}

\section{Experimental Details}
\label{app:sec:exp_details}

All experiments were run on a standard laptop/CPU machine; no GPU was used. The full experiment suite completes in under 10 minutes. 

\subsection{Experiment 1 -- Mis-selection over gap}
\label{app:subsec:experimenta}

The goal of this set of experiments is to show that TCA's mis-selection rate decays sharply with the gap $\Delta$
between the leader and the runner-up, while RR is essentially
gap-independent and saturates near $(C-1)/C$. The sweep over the number of audit targets $C$ sweep
demonstrates that the qualitative picture is robust across pool sizes
spanning 1.5 orders of magnitude; the gap required to suppress
mis-selection grows only mildly with $C$ (roughly like
$\sigma\sqrt{2\ln C}$, the scale of the maximum of $C-1$ Gaussians).

\paragraph{Setup.} We fix a single-shot configuration: each audit target has been observed for
exactly $L$ steps since its last reset, and the pending counts are $n_0 = \Delta$ for the leader (audit target $c_0)$ and $n_c = 0$ for all other audit targets.

For each $(L, \delta, \Delta)$ we run $S = 1000$ Monte-Carlo trials and
estimate $\Pr[a \neq 0]$. For TCA we exploit that the total noise on
each audit target at run length $L$ is Gaussian with variance
$\sigma^2 \sum_{k=0}^{L-1} f_k^2$, drawn independently per audit target. Greedy and uniform-random
mis-selection probabilities are computed in closed form
(0 if $\Delta>0$, $(C-1)/C$ if $\Delta=0$ for greedy; $(C-1)/C$ always
for uniform).

\begin{table*}[t]
\centering
\caption{\textbf{Experimental parameters.}}
\label{app:tab:experimental_parameters}
\begin{tabular}{ll}
\toprule
\textbf{Parameter} & \textbf{Value} \\
\midrule
audit target counts $C$ & $\{5,\;20,\;50,\;200\}$ (one figure each) \\
Run lengths $L$ & $\{100,\;1000\}$ (one panel each) \\
Privacy parameters $\delta$ & $\{0.05,\;0.10,\;0.20\}$ \\
Gaps $\Delta$ & $\{0,\;1,\;2,\;4,\;8,\;16,\;32,\;64,\;128\}$ \\
Monte-Carlo trials per point & $S = 1000$ \\
RNG seed (TCA) & 42 \\
RNG seed (RR) & 123 \\
\bottomrule
\end{tabular}
\end{table*}

\paragraph{Hyperparameters.} Calibrated noise scales are independent of $C$: TCA's $\sigma$ depends
only on $L$ and $\delta$, and likewise $p_{\mathrm{rand}}$ for RR.

\begin{table}[ht]
\centering
\caption{\textbf{Calibrated noise scales for TCA and randomized response.}}
\label{app:tab:calibrated_noise_scales}
\begin{tabular}{rrrr}
\toprule
$L$ & $\delta$ & $\sigma$ & $p_{\mathrm{rand}}$ \\
\midrule
100  & 0.05 & 12.6862 & 0.999487 \\
100  & 0.10 &  6.3306 & 0.998947 \\
100  & 0.20 &  3.1400 & 0.997771 \\
1000 & 0.05 & 14.4078 & 0.999949 \\
1000 & 0.10 &  7.1897 & 0.999895 \\
1000 & 0.20 &  3.5661 & 0.999777 \\
\bottomrule
\end{tabular}
\end{table}

\subsection{Experiment 2 -- Dynamic Utility Auditing}
\label{app:dynamic_utility_experiment}

We evaluate the mechanisms in an online report stream with report resets. The goal of this experiment is
to measure how much utility is lost over time when using a private auditing mechanism rather than the
non-private greedy baseline.

\paragraph{Mechanisms.}
We compare four auditing mechanisms:
\begin{enumerate}
    \item \textbf{Toeplitz Counter Auditing (TCA).} Each audit target runs an independent Toeplitz
    continual-counting instance. At each timestep the mechanism audits the audit target with the largest
    noisy active-count estimate. After an audit, the selected audit target's active count and private counter
    instance are reset with fresh independent randomness.

    \item \textbf{Randomized response auditing (RR).} With probability
    \(p_{\mathrm{rand}}=(1-\delta)^{1/T}\), the mechanism audits a uniformly random audit target. With
    probability \(1-p_{\mathrm{rand}}\), it audits an audit target with maximal active count, with ties broken
    uniformly at random. After the audit, the selected audit target's active count is reset. This reset step is
    important: RR is evaluated under the same active-report auditing semantics as TCA and Greedy.

    \item \textbf{Greedy auditing.} The mechanism audits an audit target with maximal active count, with
    ties broken uniformly at random. After the audit, the selected audit target's active count is reset. This
    is a non-private upper baseline.

    \item \textbf{Uniform random auditing.} The mechanism audits a uniformly random audit target,
    ignoring reports. After the audit, the selected audit target's active count is reset. This is a privacy-only
    lower baseline.
\end{enumerate}

\paragraph{Privacy calibration.}
For each run, the private mechanisms are calibrated to the same transcript-level privacy guarantee
\((0,\delta)\) over the same certified horizon \(T\). TCA uses the Gaussian Toeplitz calibration
\[
    \sigma
    =
    \frac{M_T}{2\Phi^{-1}((1+\delta)/2)},
\]
where \(M_T\) is the Toeplitz encoder sensitivity. Randomized response uses
\[
    p_{\mathrm{rand}}
    =
    (1-\delta)^{1/T}.
\]
Thus both private mechanisms are compared under the same fixed-horizon privacy requirement.

\paragraph{Report stream.}
We generate reports from an independent Poisson model. At each timestep \(t\), audit target \(0\) receives
reports
\[
    x_{t,0} \sim \mathrm{Poisson}(\lambda_{\mathrm{lead}}),
\]
and each other audit target \(c\in\{1,\ldots,C-1\}\) receives reports
\[
    x_{t,c} \sim \mathrm{Poisson}(\lambda_{\mathrm{other}}).
\]
All draws are independent across timesteps and audit targets. Although audit target \(0\) has the highest
arrival rate in expectation, the audit target with maximal active count need not be audit target \(0\) at a
given timestep, because audits reset active reports.

\paragraph{Mechanism-specific active counts.}
Each mechanism \(\mathcal M\) maintains its own active-count vector
\[
    n_t^{\mathcal M}
    =
    (n_{t, 1}^{\mathcal M},\ldots,n_{t, C}^{\mathcal M}),
\]
because different mechanisms audit different audit targets and therefore induce different reset histories.
At each timestep \(t\), the new report vector \(x_t\) is added to that mechanism's active counts. The mechanism
then selects an audited audit target \(a_t^{\mathcal M}\). All utility metrics are computed before applying
the reset. Finally, the selected audit target is reset:
\[
    n_{t,a_t^{\mathcal M}}^{\mathcal M} \leftarrow 0.
\]
For TCA, the selected audit target's private counter instance is also restarted.

\paragraph{Metrics.}
Our primary metric is the \emph{active-count deficit}
\[
    d_t^{\mathcal M}
    :=
    \max_{c\in[C]} n_{t,c}^{\mathcal M}
    -
    n_{t,a_t^{\mathcal M}}^{\mathcal M}.
\]
This measures how many fewer active reports are resolved by \(\mathcal M\)'s selected audit than by the best
available audit under the same reset history. We report the running average deficit
\[
    \bar d_t^{\mathcal M}
    :=
    \frac{1}{t}\sum_{s=1}^t d_s^{\mathcal M}.
\]
We also compute a normalized deficit,
\[
    \eta_t^{\mathcal M}
    :=
    \frac{d_t^{\mathcal M}}{\max_{c\in[C]} n_{t,c}^{\mathcal M}+1},
\]
and its running average
\[
    \bar\eta_t^{\mathcal M}
    :=
    \frac{1}{t}\sum_{s=1}^t \eta_s^{\mathcal M}.
\]
The normalized version is useful for comparing regimes with different report volumes. Finally, we record the
number of reports resolved by the selected audit,
\[
    u_t^{\mathcal M}
    :=
    n_{t, a_t^{\mathcal M}}^{\mathcal M},
\]
and its running average
\[
    \bar u_t^{\mathcal M}
    :=
    \frac{1}{t}\sum_{s=1}^t u_s^{\mathcal M}.
\]
Lower values of \(\bar d_t^{\mathcal M}\) and \(\bar\eta_t^{\mathcal M}\) indicate better performance, while
higher values of \(\bar u_t^{\mathcal M}\) indicate that audits are resolving more active reports.

\paragraph{Implementation details.}
Unless otherwise stated, we use
\[
    C=50,\qquad
    T=1000,\qquad
    \delta=0.10,\qquad
    \lambda_{\mathrm{lead}}=1.0,\qquad
    \lambda_{\mathrm{other}}=0.2.
\]
We average over \(100\) independent random seeds. For each seed, the same generated report stream is fed to
each mechanism, but each mechanism maintains its own active counts and reset history. We plot the mean
running metric across seeds with.

\paragraph{Interpretation.}
Greedy auditing has zero active-count deficit by construction. Uniform random auditing ignores reports and
therefore provides a lower utility baseline. Randomized response is expected to behave close to uniform
random auditing in long horizons, because the fixed-horizon privacy calibration makes
\(p_{\mathrm{rand}}\) close to one. TCA should accumulate substantially less active-count deficit when the
active-count gaps are large relative to the calibrated counter noise.

\section{Additional Results}
\label{app:sec:additional_results}

\subsection{Experiment 1 with \texorpdfstring{$C \in \{20,50,200\}$}{C in \{20,50,200\}}}

We vary the number of audit targets in the setup of Experiment 1 and show the results in Figures~\ref{fig:misselection_vs_gap_c20}-\ref{fig:misselection_vs_gap_c200}. TCA curves move closer to uniform random auditing, but still have low mis-selection rates as the report gap grows.

\begin{figure*}
    \centering
    \includegraphics[width=0.95\linewidth]{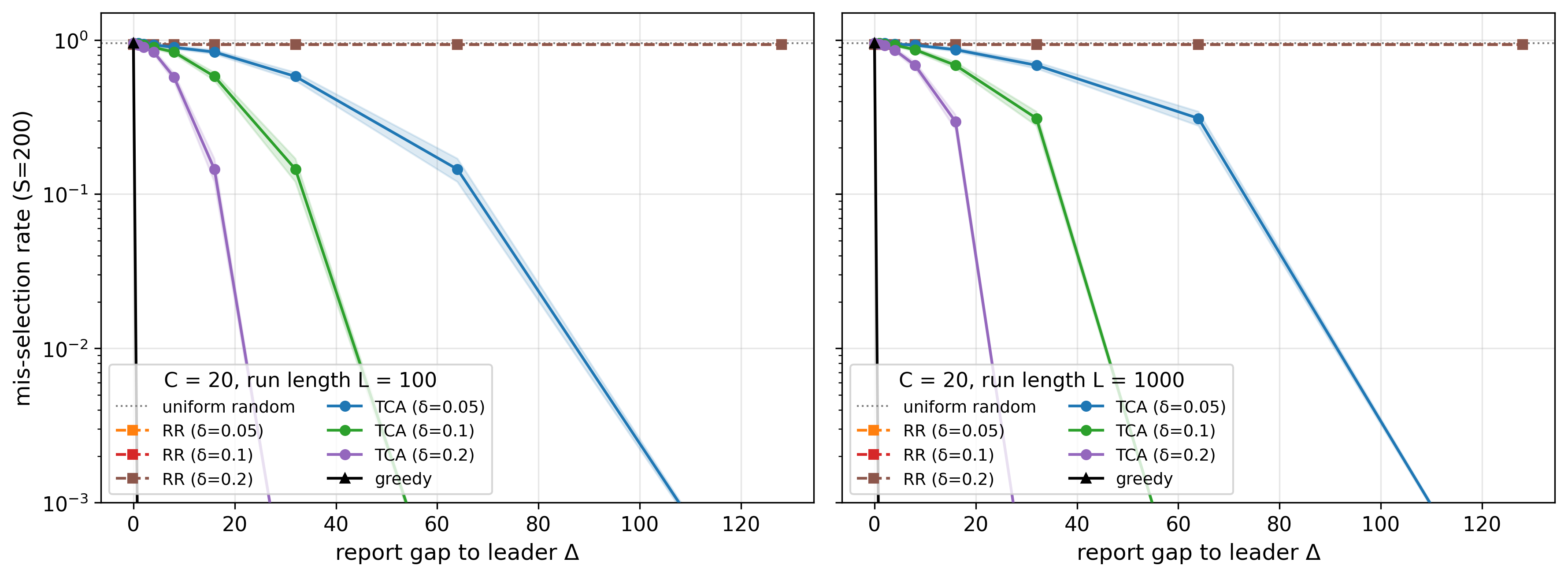}
    \caption{\textbf{Mis-selection Rate vs. Leading Gap (C=20).}}
    \label{fig:misselection_vs_gap_c20}
\end{figure*}

\begin{figure*}
    \centering
    \includegraphics[width=0.95\linewidth]{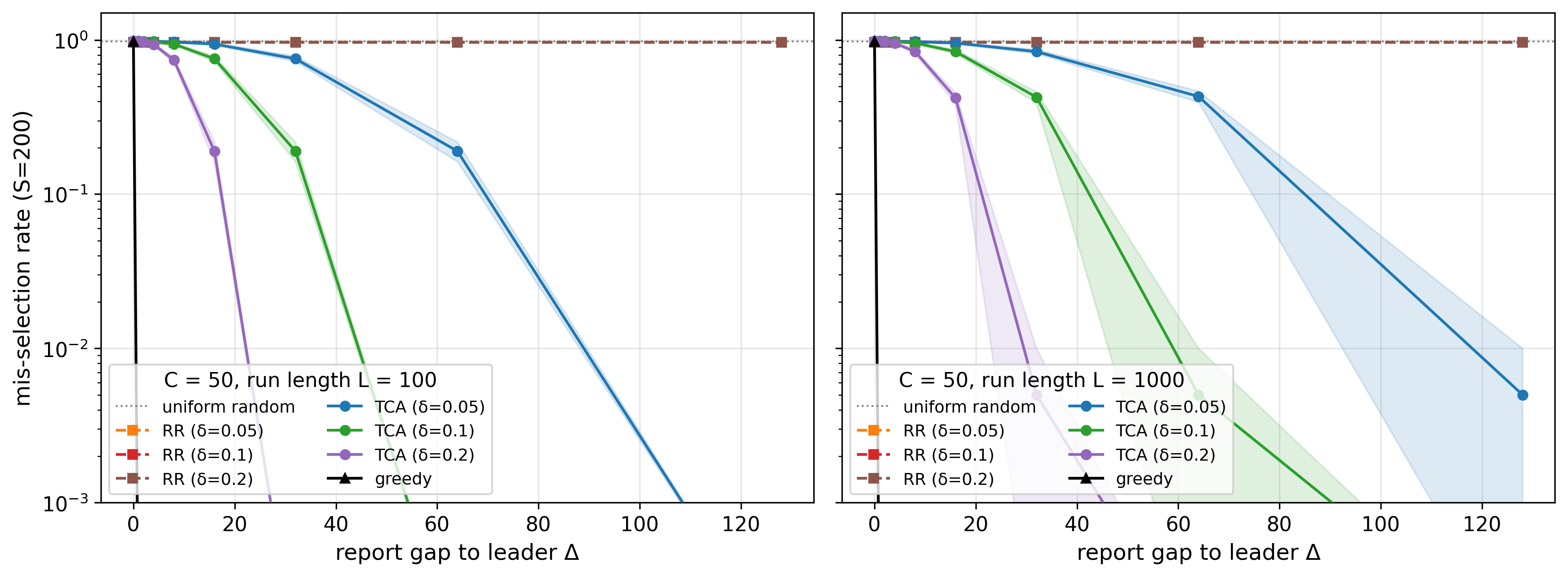}
    \caption{\textbf{Mis-selection Rate vs. Leading Gap (C=50).}}
    \label{fig:misselection_vs_gap_c50}
\end{figure*}

\begin{figure*}
    \centering
    \includegraphics[width=0.95\linewidth]{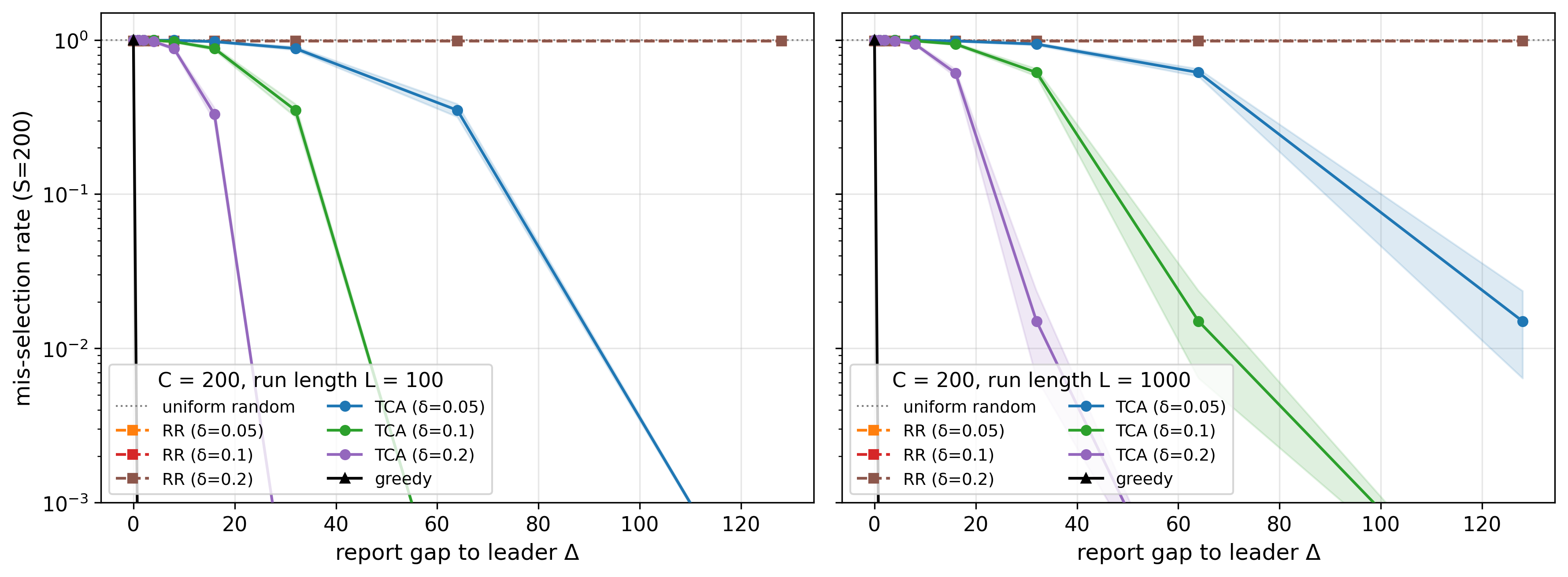}
    \caption{\textbf{Mis-selection Rate vs. Leading Gap (C=200).}}
    \label{fig:misselection_vs_gap_c200}
\end{figure*}

\section{Operationalization}
\label{app:sec:operationalization}

Our framework is defined by a small number of operational choices regarding the audit targets $c \in [C]$, the audit horizon $T$ and the target per-report deniability level $\delta$. We describe these choices in a frontier-AI auditing setting, where confidential reports from employees or contractors may inform follow-up audits, model evaluations, or compliance reviews.

In this setting, the auditor may be a public authority, an AI-safety institute, or a trusted third-party evaluator. The adversary is the audited developer, which observes the audit transcript and may know substantial internal context


The utility theorem can also be read operationally. If the true leading audit target has effective gap at least $g$ over every challenger, then Theorem~\ref{thm:utility} implies
\begin{align}
\Pr[a_t \neq c_t^* \mid \mathcal H_{t-1}]
\leq
(C-1)\Phi\!\left(-\frac{g}{\sqrt{2}\sigma}\right).
\end{align}
Thus, the main implementation trade-offs are explicit: larger $C$ requires distinguishing the leader from more challengers; smaller $\delta$ gives stronger deniability but increases noise; and larger $T$ has only a logarithmic effect on the noise scale.

\end{document}